\documentclass[orivec]{llncs}

\usepackage{microtype}\usepackage[utf8]{inputenc}
\usepackage{amsmath}
\usepackage{tikz}
\usetikzlibrary{arrows,shapes,snakes,automata,backgrounds,
                petri,positioning,fit}
\usepackage[font={small}, margin=20pt]{caption}
\usepackage[]{algorithm}
\usepackage[noend]{algpseudocode}
\usepackage{enumerate}
\usepackage{bm}  \usepackage{wrapfig}
\usepackage{hyperref}

\usepackage{amsmath}
\usepackage{paralist}
\usepackage{mathtools}

\usepackage{color}

\definecolor{lightblue}{RGB}{231,255,255}
\definecolor{lightred}{RGB}{255,224,224}
\definecolor{lightgreen}{RGB}{224,255,224}
\definecolor{lightyellow}{RGB}{255,255,224}
\definecolor{lightpurple}{RGB}{255,224,255}
\definecolor{darkerred}{RGB}{64,0,0}
\definecolor{darkred}{RGB}{128,0,0}
\definecolor{darkblue}{RGB}{0,0,128}
\definecolor{darkgreen}{RGB}{0,128,0}
\definecolor{darkpurple}{RGB}{128,0,128}
\definecolor{black}{RGB}{0,0,0}

\newcommand{\clocks}{\mathcal{C}}
\newcommand{\actions}{\mathcal{A}}
\newcommand{\inactions}{\actions^{\mathsf i}}
\newcommand{\outactions}{\actions^{\mathsf o}}
\newcommand{\coactions}{\actions^{\mathsf u}}
\newcommand{\states}{\mathcal{S}}
\newcommand{\trans}[1][]{\xrightarrow{{#1}}}

\newcommand{\dedrule}[2]{\frac{#1}{#2}}
\newcommand{\iosa}{IOSA}
\newcommand{\iosac}{\iosa}

\newcommand{\I}{\mathcal{I}}
\newcommand{\pll}{||}

\newcommand{\act}{\mathit{a}}
\newcommand{\actb}{\mathit{b}}
\newcommand{\init}{\mathsf{init}}
\newcommand{\sinit}{\mathsf{init}}

\newcommand{\E}{\mathrm{E}}
\newcommand{\reduct}[1][]{\overset{_{#1}}{\rightarrowtail}}
\newcommand{\invreduct}[1][]{\overset{_{#1}}{\leftarrowtail}}

\newcommand{\II}{\mathcal{I}}

\newcommand{\uars}{\mathcal{U}}

\newcommand{\safe}[1][]{\ifthenelse{\equal{#1}{}}{\mathrm{safe}}{\mathrm{safe\hspace{1pt}}(#1)}}

\newcommand{\emptyfun}{\ensuremath\mathbf{0}}

\newcommand{\stable}[1]{\mathrm{st\hspace{1pt}}(#1)}

\newcommand{\istates}{\mathcal{S}}

\newcommand{\Trans}[1][]{\xRightarrow{\, {#1} \, }}

\makeatletter
\def\THICKhrulefill{\leavevmode \leaders \hrule height 5pt\hfill \kern \z@}
\makeatother

\newcommand{\activeck}{\ensuremath{\mathrm{active}}}
\newcommand{\enablingck}{\ensuremath{\mathrm{enabling}}}
\newcommand{\inv}{\mathsf{Inv}}
\newcommand{\notinv}{\mathsf{Inv}^c}

\usepackage{paralist}
\usepackage{amssymb}
\usepackage{mathrsfs}
\usepackage{xifthen}

\newcommand{\pstates}{\textbf{S}}
\newcommand{\labels}{\mathcal{L}}
\newcommand{\transitions}{\mathcal{T}}
\newcommand{\borel}[1][]{\ifthenelse{\equal{#1}{}}{\mathscr{B}}{\mathscr{B}(#1)}}

\newcommand{\sampleSp}{S}
\renewcommand{\emptyset}{\varnothing}

\renewcommand{\P}{\mathcal{P}}

\newcommand{\nlmp}{NLMP}

\renewcommand{\phi}{\varphi}

\renewcommand{\iff}{\Leftrightarrow}

\newcommand{\sigmaA}{\Sigma}
\newcommand{\R}{\ensuremath{\mathbb{R}}}

\newcommand{\N}{\ensuremath{\mathbb{N}}}
\newcommand{\pathWrite}{\mathit{Path}}

\newcommand{\Path}[1]{
  \ifthenelse{\equal{#1}{}}
  {\ensuremath{\pathWrite}}
  {\ensuremath{\pathWrite(#1)}}
}
\newcommand{\finpath}[1]{
  \ifthenelse{\equal{#1}{}}
  {\ensuremath{\pathWrite^*}}
  {\ensuremath{\pathWrite^*(#1)}}
}
\newcommand{\omegapath}[1]{
  \ifthenelse{\equal{#1}{}}
  {\ensuremath{\pathWrite^\omega}}
  {\ensuremath{\pathWrite^\omega(#1)}}
}

\newcommand{\bisim}{\sim}

\newcommand{\partsof}{\mathcal{P}}

\newcommand{\compactcdots}{{\cdot}{\cdot}{\cdot}}
\newcommand{\compactints}{\int\!\!\compactcdots\!\!\int}

\newcommand{\uen}{\ensuremath{\mathrm{uen}}}
\newcommand{\ES}{\ensuremath{\mathsf{ES}}}
\newcommand{\EG}{\ensuremath{\mathsf{EG}}}
\newcommand{\triggers}{\leadsto}

\title{Input/Output Stochastic Automata with Urgency: Confluence and
  weak determinism\thanks{This work was supported by grants ANPCyT
    PICT-2017-3894 (RAFTSys), SeCyT-UNC 33620180100354CB (ARES), and
    the ERC Advanced Grant 695614 (POWVER).}}
\titlerunning{Input/Output Stochastic Automata with Urgency}

\author{Pedro R. D'Argenio\inst{1,2,3}, Ra\'ul E. Monti\inst{1,2}}
\institute{
Universidad Nacional de C\'ordoba, FAMAF, C\'ordoba, Argentina
\and CONICET, C\'ordoba, Argentina
\and Saarland University, Department of Computer Science, Saarbr\"ucken, Germany}

\authorrunning{P.\,R. D'Argenio and R.\,E. Monti}

\begin{document}

\maketitle

\begin{abstract}
  In a previous work, we introduced an input/output variant of
  stochastic automata (\iosa) that, once the model is closed (i.e., all
  synchronizations are resolved), the resulting automaton is fully
  stochastic, that is, it does not contain non-deterministic choices.
  However, such variant is not sufficiently versatile for
  compositional modelling.
  In this article, we extend \iosa\ with urgent actions.  This extension
  greatly increases the modularization of the models, allowing to take
  better advantage on compositionality than its predecessor.  However,
  this extension introduces non-determinism even in closed models.  We
  first show that confluent models are weakly deterministic in the
  sense that, regardless the resolution of the non-determinism, the
  stochastic behaviour is the same.  In addition, we provide
  sufficient conditions to ensure that a network of interacting \iosa{s}
  is confluent without the need to analyse the larger composed \iosa.
\end{abstract}

\section{Introduction}\label{sec:introduction}

The advantages of compositional modelling complex systems can hardly be
overestimated.
On the one hand, compositional modelling facilitates systematic
design, allowing the designer to focus on the construction of small
models for the components whose operational behavior is mostly well
understood, and on the synchronization between the components, which
are in general quite evident.
On the other hand, it facilitates the interchange of components in a
model, enables compositional analysis, and helps on attacking the state
explosion problem.

In particular we focus on modelling of stochastic system for
dependability and performance analysis, and aim to general models that
require more than the usual negative exponential distribution.
Indeed, phenomena such as timeouts in communication protocols, hard
deadlines in real-time systems, human response times or the
variability of the delay of sound and video frames (so-called jitter)
in modern multi-media communication systems are typically described by
non-memoryless distributions such as uniform, log-normal, or Weibull
distributions.

The analysis of this type of model quite often can only be performed
through discrete event simulation~\cite{LawKelton:1999}.  However,
simulation requires that the model under study is fully stochastic,
that is, they should not contain non-deterministic choices.
Unfortunately, compositional modelling languages such as stochastic process
algebras with general distributions (see~\cite{bravettiDArgenio04} and
references therein) and
Modest~\cite{tse/BohnenkampDHK06,fmsd/HahnHHK13,Hartmanns15:phd}, were
designed so that the non-determinism arises naturally as the result of
composition.

Based on stochastic
automata~\cite{procomet/DArgenioKB98,DArgenio99:phd,DArgenioK05:iandc}
and probabilistic I/O automata~\cite{WuSS97:tcs}, we introduced
input/output stochastic automata (\iosa)~\cite{DArgenio2016}.
\iosa{s} were designed so that parallel composition works naturally
and, moreover, the system becomes fully stochastic --not containing
non-determinism-- when closed, i.e., when all interactions are
resolved and no input is left available in the model.
\iosa\ splits the set of actions into inputs and outputs and let
them behave in a reactive and generative manner
respectively~\cite{GlabbeekSS95:iandc}.  Thus, inputs are passive and
their occurrence depends only on their interaction with outputs.
Instead, occurrence of outputs are governed by the expiration of a timer
which is set according to a given random variable.
In addition, and not to block the occurrence of outputs, \iosa{s} are
required to be input enabled.

\begin{wrapfigure}[8]{r}{0.39\textwidth}
\centering\vspace{-2em}
\includegraphics[scale=.4]{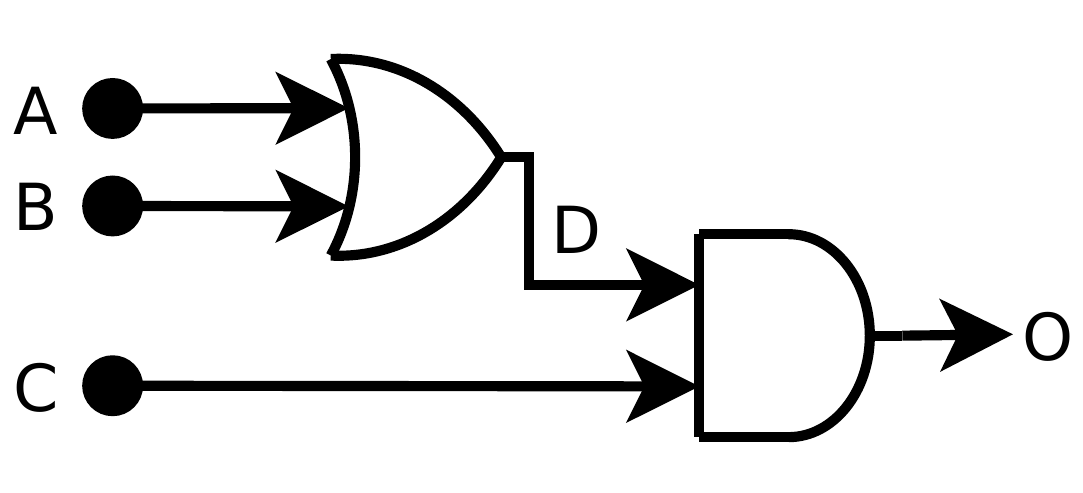}
\caption{A simple digital system.}
\label{fig:digital}
\end{wrapfigure}
We have used \iosa\ as input language of the rare event simulation
tool FIG~\cite{BuddeDM2016,Budde2017:phd} and have experienced the
limitations of the language, in particular when transcribing models
originally given in terms of variants of dynamic fault trees (DFT)
with repairs~\cite{DBLP:journals/csr/RuijtersS15}.
To illustrate the problem, suppose the simple digital system of
Fig.~\ref{fig:digital}.  We would like to measure the average time
that the output \textsf{O} is 1 given that we know the distributions
of the times in which the values on inputs \textsf{A}, \textsf{B}, and
\textsf{C} change from 0 to 1 and vice-versa.  The natural modelling of such system is to define
5 \iosa\ modules, three of them modelling the behaviour of the input
signals and the other two modelling the OR and AND gates.  Then we compose and
synchronize the 5 modules properly.  The main problem is that, while
the dynamic behaviour of the input signal modules are governed by
stochastically timed actions, the dynamic behavior of the gates are
instantaneous and thus, for instance the output \textsf{D} of the OR
gate, may change immediately after the arrival of signals
\textsf{A} or \textsf{B}.
Similar situations arise when modeling the behaviour of DFT under
complex gates like priority AND, Spares or Repair boxes.
As a consequence, we observe that the introduction of urgent actions
will allow for a direct and simple compositional modelling of
situations like the one recently described.
Also, it is worth to notice that the need for instantaneous but
causally dependent synchronization have been observed in many other
timed modelling languages, notably, in Uppaal, with the introduction
of committed locations, urgent locations and urgent
synchronization~\cite{DBLP:conf/cav/BengtssonGKLLPY96,DBLP:conf/sfm/BehrmannDL04}

Based on IMC~\cite{Hermanns02} and, particularly, on
I/O-IMC~\cite{Crouzen2014:phd}, in this article we extended
\iosa\ with urgent actions (Sec.~\ref{sec:iosa}).  Urgent actions are
also partitioned in input and output actions and, though inputs behave
reactively and passively as before, urgent outputs are executed
instantaneously as soon as the enabling state is reached.
We also give semantics to \iosa\ with urgent actions (from now
on, we simply call it \iosa) in terms of
NLMP~\cite{DArgenioSTW12:mscs,Wolovick2012:phd}
(Sec.~\ref{sec:semantics}), and define its parallel composition
(Sec.~\ref{sec:parallelcomposition}.)

The problem is that urgent actions on \iosa\ introduce non-determinism.
Fortunately, non-determinism is limited to urgent actions and, in many
occasions, it is introduced by confluent urgent output actions as a
result of a parallel composition.  Such non-determinism turns to be
spurious in the sense that it does not change the stochastic behaviour
of the model.
In this paper, we characterize confluence on \iosa{s}
(Sec.~\ref{sec:confluence}), define the concept of weak determinism,
and show that a confluent closed \iosa\ is weakly deterministic
(Sec.~\ref{sec:weak:det}).  Notably, a weakly deterministic \iosa\ is
amenable to discrete event simulation.
Milner~\cite{Milner1989:book} has provided a proof that confluence
preserves weak determinism but it is confined to a discrete
non-probabilistic setting.  A similar proof has been used by
Crouzen~\cite{Crouzen2014:phd} on I/O-IMC but, though the model is
stochastic, the proof is limited to discrete non-probabilistic
transitions.  Contrarily, our proof has to deal with continuous
probabilities (since urgent action may sample on continuous random
variables), hence making use of the solid measure theoretical
approach.  In particular, we address the complications of defining a
particular form of weak transition on a setting that is normally
elusive.

Based on the work of Crouzen~\cite{Crouzen2014:phd} for I/O-IMC, in
Sec.~\ref{sec:ensuring-confluence}, we provide sufficient conditions
to ensure that a closed \iosa\ is confluent and hence, weakly
deterministic.  If the \iosa\ is the result of composing several
smaller \iosa{s}, the verification of the conditions is performed by
inspecting the components rather than the resulting composed \iosa.

\section{ Input/Output Stochastic Automata with urgency.}
\label{sec:iosa}

Stochastic automata \cite{DArgenio99:phd,DArgenioK05:iandc} use
continuous random variables (called \emph{clocks}) to observe the passage
of time and control the occurrence of events. These variables are
set to a value according to their associated probability
distribution, and, as time evolves, they count down at the same
rate. When a clock reaches zero, it may trigger some action.
This allows the modelling of systems where events occur at random
continuous time steps.

Following ideas from \cite{WuSS97:tcs}, \iosa{s} restrict Stochastic
Automata by splitting actions into input and output actions which
will act in a reactive and generative way
respectively~\cite{GlabbeekSS95:iandc}.  This splitting reflects the fact that
input actions are considered to be controlled externally, while
output actions are locally controlled.

Therefore, we consider the system to be input enabled. Moreover,
output actions could be stochastically controlled or
instantaneous. In the first case, output actions are controlled
by the expiration of a single clock while in the second case the
output actions take place as soon as the enabling state is
reached. We called these instantaneous actions \emph{urgent}. A
set of restrictions over \iosa\ will ensure that, almost surely,
no two non-urgent outputs are enabled at the same time.

\begin{definition}\label{def:iosac}
  An \emph{input/output stochastic automaton with urgency} (\iosa) is
  a structure $(\states,\actions,\clocks,\trans,C_0,s_0)$, where
  $\states$ is a (denumerable) set of states, $\actions$ is a
  (denumerable) set of labels partitioned into disjoint sets of
  \emph{input} labels $\inactions$ and \emph{output} labels
  $\outactions$, from which a subset $\coactions\subseteq\actions$ is
  marked as \emph{urgent}.  We consider the distinguished \emph{silent}
  urgent action $\tau\in {\coactions\cap\outactions}$ which is not amenable to
  synchronization.  $\clocks$ is a (finite) set of clocks such that
  each $x\in\clocks$ has an associated continuous probability measure
  $\mu_x$ on $\R$ s.t.\ $\mu_x(\R_{>0})=1$,
  ${\trans}\subseteq\states\times\clocks\times\actions
  \times\clocks\times S$ is a transition function, $C_0$ is the set of
  clocks that are initialized in the initial state, and $s_0\in
  \states$ is the initial state.

  In addition, an \iosa\ with urgency should satisfy the following
  constraints:
  \begin{enumerate}[\rm a]
    \renewcommand{\theenumi}{(\alph{enumi})}
  \item\label{def:iosa:input-and-commit-are-reactive}    If $s\trans[C,\act,C']s'$ and $\act\in \inactions\cup\coactions$,
    then $C=\emptyset$.
  \item\label{def:iosa:output-is-generative}    If $s\trans[C,\act,C']s'$ and $\act\in \outactions\setminus\coactions$,
    then $C$ is a singleton set.
  \item\label{def:iosa:clock-control-one-output}    If $s\trans[\lbrace x\rbrace,\act_1,C_1]s_1$ and
    $s\trans[\lbrace x\rbrace,\act_2,C_2]s_2$
    then $\act_1=\act_2$, $C_1=C_2$ and $s_1=s_2$.
  \item\label{def:iosa:input-enabled}    For every $\act\in\inactions$ and state $s$, there exists a
    transition $s\trans[\emptyset,\act,C]s'$.
  \item\label{def:iosa:input-deterministic}    For every $\act\in\inactions$, if $s\trans[\emptyset,\act,C_1']s_1$ and
    $s\trans[\emptyset,\act,C_2']s_2$, $C_1' = C_2'$ and $s_1 = s_2$.
  \item\label{def:iosa:clock-never-go-off-early}    There exists a function $\activeck:\states\rightarrow 2^{\clocks}$ such that:
    \begin{inparaenum}[i]
    \renewcommand{\theenumii}{\normalfont(\roman{enumii})}
    \item\label{def:iosa:clock-never-go-off-early:i}      $\activeck(s_0)\subseteq C_0$,
          \item\label{def:iosa:clock-never-go-off-early:ii}      $\enablingck(s)\subseteq \activeck(s)$,
          \item\label{def:iosa:clock-never-go-off-early:iii}      if $s$ is stable, $\activeck(s)=\enablingck(s)$, and
          \item\label{def:iosa:clock-never-go-off-early:iv}      if $t\trans[C,a,C']s$ then
            $\activeck(s)\subseteq(\activeck(t)\setminus C)\cup C'$.
    \end{inparaenum}
  \end{enumerate}
  where $\enablingck(s)=\{y\mid s\trans[\{y\},\_,\_]\_\}$, and $s$ is
  \emph{stable}, denoted $\stable{s}$, if there is no
  $a\in\coactions\cap\outactions$ such that
  $s\trans[\emptyset,a,\_]\_$. ($\_$ indicates the existential
  quantification of a parameter.)
\end{definition}

The occurrence of an output transition is controlled by the expiration
of clocks.  If $a\in\outactions$, $s\trans[C,a,C']s'$ indicates that
there is a transition from state $s$ to state $s'$ that can be taken
only when all clocks in $C$ have expired and, when taken, it triggers
action $a$ and sets all clocks in $C'$ to a value sampled from their
associated probability distribution.  Notice that if $C=\emptyset$
(which means $a\in\outactions\cap\coactions$) $s\trans[C,a,C']s'$ is
immediately triggered.
Instead, if $a\in\inactions$, $s\trans[\emptyset,a,C']s'$ is only
intended to take place if an external output synchronizes with it,
which means, in terms of an open system semantics, that it may take
place at any possible time.

Restrictions \ref{def:iosa:input-and-commit-are-reactive} to
\ref{def:iosa:input-deterministic} ensure that any \emph{closed}
\iosa\ without urgent actions is deterministic~\cite{DArgenio2016}.
An \iosa\ is closed if all its synchronizations have been resolved,
that is, the \iosa\ resulting from a composition does not have input
actions ($\inactions=\emptyset$).
Restriction \ref{def:iosa:input-and-commit-are-reactive} is
two-folded: on the one hand, it specifies that output actions must
occur as soon as the enabling state is reached, on the other hand,
since input actions are reactive and their time occurrence can only
depend on the interaction with an output, no clock can control their
enabling.  Restriction \ref{def:iosa:output-is-generative} specifies
that the occurrence of a non-urgent output is locally controlled by a
single clock.  Restriction \ref{def:iosa:clock-control-one-output}
ensures that two different non-urgent output actions leaving the same
state are always controlled by different clocks (otherwise it would
introduce non-determinism).  Restriction \ref{def:iosa:input-enabled}
ensures input enabling.  Restriction
\ref{def:iosa:input-deterministic} determines that \iosa{s} are input
deterministic. Therefore, the same input action in the same state can
not jump to different states, nor set different clocks.
Finally, \ref{def:iosa:clock-never-go-off-early}~guarantees that clocks
enabling some output transition have not expired before, that is, they
have not been used before by another output transition (without being
reset in between) nor inadvertently reached zero.
This is done by ensuring the existence of a function ``$\activeck$''
that, at each state, collects clocks that are required to be active
(i.e. that have been set but not yet expired).  Notice that enabling
clocks are required to be active
(conditions~\ref{def:iosa:clock-never-go-off-early:ii}
and~\ref{def:iosa:clock-never-go-off-early:iii}). Also note that every
clock that is active in a state is allowed to remain active in a
successor state as long as it has not been used, and clocks that have
just been set may become active in the successor state
(condition~\ref{def:iosa:clock-never-go-off-early:iv}).

Note that since clocks are set by sampling from a continuous random
variable, the probability that the values
of two different clocks are equal is 0. This fact along with restriction
\ref{def:iosa:clock-control-one-output} and
\ref{def:iosa:clock-never-go-off-early}
guarantee that almost never two different non-urgent
output transitions are enabled at the same time.

\begin{wrapfigure}[10]{r}{0.36\textwidth}
    \vspace{-2.7em}
  \centering
  \noindent
  \begin{tikzpicture}[shorten >=1pt,node distance=2.2cm,on grid,auto, scale=0.7, transform shape, initial text={}, initial where=above]
    \tikzstyle{every state}=[text=black,draw=black,minimum size=.5cm]
    \node[state,initial]    (s1)                 {$s0$};
    \node[state]            (s2)   [right=of s1] {$s1$};
    \node[state]            (s3)   [right=of s2] {$s2$};
    \node[state,draw=white] (I1) [above left =.7cm and .5cm of s1] {$\II_1$};
    \path[->]
    (s1) edge []       node {$\{x\},a!,\emptyset$} (s2)
    (s2) edge []       node {$\emptyset,c!!,\emptyset$}      (s3);
  \end{tikzpicture}

  \noindent
  \begin{tikzpicture}[shorten >=1pt,node distance=2.2cm,on grid,auto, scale=0.7, transform shape, initial text={}, initial where=above]
    \tikzstyle{every state}=[text=black,draw=black,minimum size=.5cm]
    \node[state,initial]    (s1)                {$s3$};
    \node[state]            (s2) [right=of s1]  {$s4$};
    \node[state]            (s3) [right=of s2]  {$s5$};
    \node[state,draw=white] (I2) [above left =.7cm and .5cm of s1] {$\II_2$};
    \path[->]
    (s1) edge []       node {$\{y\},b!,\emptyset$} (s2)
    (s2) edge []       node {$\emptyset,d!!,\emptyset$}      (s3);
  \end{tikzpicture}

  \medskip

  \begin{tikzpicture}[shorten >=1pt,node distance=2.2cm,on grid,auto,
      scale=0.7, transform shape, initial text={}, initial where=above]
    \tikzstyle{every state}=[text=black,draw=black,minimum size=.5cm]
    \node[state,initial] (s1)                               {$s6$};
    \node[state] (s2) [above right =1cm and 2.2cm of s1] {$s7$};
    \node[state] (s3) [below right =1cm and 2.2cm of s2] {$s8$};
    \node[state] (s4) [below right =1cm and 2.2cm of s1] {$s9$};
    \node[state,draw=white] (I2) [above left =1cm and .5cm of s1] {$\II_3$};
    \path[->]
    (s1) edge [] node[midway,above,rotate=26.48] {$\emptyset,c??,\emptyset$} (s2)
    (s1) edge [] node[midway,below,rotate=-26.48] {$\emptyset,d??,\emptyset$} (s4)
    (s2) edge [] node[midway,above,rotate=-26.48] {$\emptyset,d??,\{z\}$} (s3)
    (s3) edge [] node[midway,below,rotate=26.48] {$\{z\},e!,\emptyset$} (s4);
  \end{tikzpicture}

  \caption{Examples of \iosa{s}.}
  \label{fig:components}
\end{wrapfigure}

\medskip
\noindent\textit{Example \refstepcounter{example}\label{ex:iosas}\theexample.}
  Fig.~\ref{fig:components} depicts three simple examples of \iosa{s}.
  Although \iosa{s} are input enabled, we have omitted self loops of
  input enabling transitions for the sake of readability.
  In the figure, we represent output actions suffixed by `!' and by
  `!!' when they are urgent, and input actions suffixed by `?' and by
  `??' when they are urgent.

\section{Semantics of \iosa}\label{sec:semantics}

The semantics of \iosa\ is defined in terms of non-deterministic
labeled Markov processes
(\nlmp)~\cite{DArgenioSTW12:mscs,Wolovick2012:phd} which extends
LMP~\cite{DesharnaisEP02:iandc} with \emph{internal} non-determinism.

The foundations of \nlmp\ is strongly rooted in measure theory, hence
we recall first some basic definitions.  Given a set $\sampleSp$ and a
collection $\sigmaA$ of subsets of $\sampleSp$, we call $\sigmaA$ a
\emph{$\sigma$-algebra} iff $\sampleSp \in \sigmaA$ and $\sigmaA$ is
closed under complement and denumerable union.  We call the pair
$(\sampleSp, \sigmaA)$ a \emph{measurable space}.  Let
$\borel(\sampleSp)$ denote the Borel $\sigma$-algebra on the topology
$\sampleSp$.
A function $\mu:\sigmaA\to [0,1]$ is a \emph{probability measure} if
\begin{inparaenum}[(i)]
\item  $\mu(\bigcup_{i\in\N} Q_i)=\sum_{i\in\N} \mu(Q_i)$ for all countable
  family of pairwise disjoint measurable sets $\{Q_i\}_{i\in\N}
  \subseteq \sigmaA$, and
\item
  $\mu(\sampleSp)=1$.
\end{inparaenum}
In particular, for $s\in S$, $\delta_s$ denotes the Dirac measure so
that $\delta_s(\{s\})=1$.
Let $\Delta(\sampleSp)$ denote the set of all probability measures
over $(\sampleSp,\sigmaA)$.
Let $(\sampleSp_1, \sigmaA_1)$ and $(\sampleSp_2, \sigmaA_2)$ be two
measurable spaces.  A function $f:\sampleSp_1\to\sampleSp_2$ is said
to be \emph{measurable} if for all $Q_2\in\sigmaA_2$,
$f^{-1}(Q_2)\in\sigmaA_1$.
There is a standard construction to endow
$\Delta(\sampleSp)$ with a $\sigma$-algebra \cite{giry81catprob} as follows:
$\Delta(\sigmaA)$ is defined as the smallest $\sigma$-algebra
containing the sets $\Delta^q(Q) \doteq \{\mu \mid \mu(Q)\geq q \}$,
with $Q\in \sigmaA$ and $q\in[0,1]$.
Finally, we define the \emph{hit $\sigma$-algebra} $H(\Delta(\Sigma))$
as the minimal $\sigma$-algebra containing all sets $H_{\xi} =
\{\zeta\in\Delta(\Sigma) \mid \zeta\cap\xi \neq\emptyset \}$ with
$\xi\in\Delta(\Sigma)$.

  A \emph{non-deterministic labeled Markov process} (\nlmp\ for short)
  is a structure
  $(\pstates,\Sigma,\{\transitions_\act\mid\act\in\labels\})$ where
  $\Sigma$ is a $\sigma$-algebra on the set of states $\pstates$, and
  for each label $\act\in\labels$ we have that $\transitions_\act
  :\pstates\to\Delta(\Sigma)$ is measurable from $\Sigma$ to
  $H(\Delta(\Sigma))$.

The formal semantics of an \iosa\ is defined by a \nlmp\ with two
classes of transitions: one that encodes the discrete steps and
contains all the probabilistic information introduced by the sampling
of clocks, and another describing the time steps, that only records
the passage of time synchronously decreasing the value of all clocks.
For simplicity, we assume that the set of clocks has a total order and
their current values follow the same order in a vector.

\begin{definition}\label{def:iosasemantic}
  Given an \iosa\ $\I =
  (\istates,\actions,\clocks,\trans,C_0,s_0)$ with $\clocks =
  \{x_1, \ldots, x_N\}$, its semantics is defined by the
  \nlmp\ $\P(\I)
  =(\pstates,\borel(\pstates),\{\transitions_\act\mid\act\in\labels\})$
  where
  \begin{itemize}
  \item    $\pstates = (\istates\cup\{\init\}) \times\R^N$,
    $\labels = \actions \cup \R_{>0} \cup \{\init\}$, with
    $\init\notin\istates\cup\actions\cup\R_{>0}$
  \item    $\transitions_\init(\init,\vec{v}) =
    \{\delta_{s_0}\times\prod_{i=1}^N\mu_{x_i}\}$,
  \item    $\transitions_\act(s,\vec{v}) = \{\mu^{\vec{v}}_{ C',s'} \mid
    s\trans[C, \act, C']s', \bigwedge_{x_i \in C} \vec{v}(i)\leq 0 \}$,
    for all $a\in\actions$, where
    $\mu^{\vec{v}}_{C', s'} = \delta_{s'} \times \prod_{i=1}^N \overline{\mu}_{x_i}$
    with $\overline{\mu}_{x_i} = {\mu}_{x_i}$ if
    $x_i \in C'$ and $\overline{\mu}_{x_i} = \delta_{\vec{v}(i)}$ otherwise, and
  \item    $\transitions_d(s,\vec{v}) = \{\delta_{s} \times \prod_{i=1}^N\delta_{\vec{v}(i)-d}\}$
    if there is no urgent $b\in\outactions\cap\coactions$
    for which $s\trans[\_,b,\_]\_$ and
    $0 < d \leq \min\{\vec{v}(i) \mid
        \exists a{\in}\outactions, C'{\subseteq}\clocks, s'{\in}S :
           s\trans[\{x_i\}, a, C'] s' \}$,
    and
    $\transitions_d(s,\vec{v}) = \emptyset$ otherwise, for all $d\in\R_{\geq 0}$.
  \end{itemize}
\end{definition}

The state space is the product space of the states of the \iosa\ with
all possible clock valuations. A distinguished initial state $\init$
is added to encode the random initialization of all clocks (it would
be sufficient to initialize clocks in $C_0$ but we decided for this
simplification).  Such encoding is done by transition
$\transitions_\init$.  The state space is structured with the usual
Borel $\sigma$-algebra.  The discrete step is encoded by
$\transitions_\act$, with $a\in\actions$.  Notice that, at state
$(s,\vec{v})$, the transition $s\trans[C, \act, C']s'$ will only take
place if $\bigwedge_{x_i \in C} \vec{v}(i)\leq 0$, that is, if the
current values of all clocks in $C$ are not positive.  For the
particular case of the input or urgent actions this will always be
true.  The next actual state would be determined randomly as follows:
the symbolic state will be $s'$ (this corresponds to $\delta_{s'}$ in
$\mu^{\vec{v}}_{C', s'} = \delta_{s'} \times
\prod_{i=1}^N\overline{\mu}_{x_i}$), any clock not in $C'$ preserves
the current value (hence $\overline{\mu}_{x_i} = \delta_{\vec{v}(i)}$
if $x_i\notin C'$), and any clock in $C'$ is set randomly according to
its respective associated distribution (hence $\overline{\mu}_{x_i} =
{\mu}_{x_i}$ if $x_i \in C'$).  The time step is encoded by
$\transitions_d(s,\vec{v})$ with $d\in\R_{\geq 0}$.  It can only take
place at $d$ units of time if there is no output transition enabled at
the current state within the next $d$ time units (this is verified by
condition $0 < d \leq \min\{\vec{v}(i) \mid \exists a{\in}\outactions,
C'{\subseteq}\clocks, s'{\in}S : s\trans[\{x_i\}, a, C'] s' \}$).  In
this case, the system remains in the same symbolic state (this
corresponds to $\delta_{s}$ in
$\delta^{-d}_{(s,\vec{v})}=\delta_{s}\times\prod_{i=1}^N\delta_{\vec{v}(i)-d}$),
and all clock values are decreased by $d$ units of time (represented
by $\delta_{\vec{v}(i)-d}$ in the same formula).  Note the difference
from the timed transitions semantics of pure
\iosa~\cite{DArgenio2016}. This is due to the maximal progress
assumption, which forces to take urgent transition as soon as they get
enabled.  We encode this by not allowing to make time transitions in
presence of urgent actions, i.e.\ we check that there is no urgent
$b\in\outactions\cap\coactions$ for which $s\trans[\_,b,\_]\_$.
(Notice that $b$ may be $\tau$.)  Otherwise,
$\transitions_d(s,\vec{v})=\emptyset$.  Instead, notice the
\emph{patient} nature of a state $(s,\vec{v})$ that has no output
enabled.  That is, $\transitions_d(s,\vec{v}) = \{\delta_{s} \times
\prod_{i=1}^N\delta_{\vec{v}(i)-d}\}$ for all $d>0$ whenever there is
no output action $b\in\outactions$ such that $s\trans[\_,b,\_]\_$.

In a similar way to
\cite{DArgenio2016}, it is possible to show that $\P(\I)$ is indeed a
\nlmp, i.e. that $\transitions_a$ maps into measurable sets in
$\Delta(\borel(\pstates))$, and that $\transitions_a$ is a measurable
function for every $a\in\labels$.

\section{Parallel Composition}\label{sec:parallelcomposition}

In this section, we define parallel composition of \iosa{s}. Since
outputs are intended to be autonomous (or locally controlled), we do
not allow synchronization between them. Besides, we need to avoid name
clashes on the clocks, so that the intended behavior of each component
is preserved and moreover, to ensure that the resulting composed
automaton is indeed an \iosa.  Furthermore, synchronizing \iosa{s}
should agree on urgent actions in order to ensure their immediate
occurrence. Thus we require to compose only \emph{compatible}
\iosa{s}.

\begin{definition}
  \label{def:compatible}
  Two \iosac{s} $\II_1$ and $\II_2$ are \emph{compatible} if they do
  not share synchronizable output actions nor clocks, i.e.
  $\outactions_1 \cap \outactions_2 \subseteq\{\tau\}$ and
  $\clocks_1\cap\clocks_2 = \emptyset$ and, moreover, they agree on
  urgent actions, i.e.\
  $\actions_1\cap\coactions_2 = \actions_2\cap\coactions_1$.
\end{definition}

\begin{definition}
  \label{def:parcomp}
  Given two compatible \iosac{s} $\II_1$ and $\II_2$, the parallel
  composition $\II_1\pll \II_2$ is a new
  \iosac\ $(\states_1\times\states_2, \actions, \clocks, \trans,
  C_0, s_0^1\pll s_0^2)$ where
  \begin{inparaenum}[(i)]
  \item    $\outactions = \outactions_1 \cup \outactions_2$
  \item    $\inactions = (\inactions_1 \cup \inactions_2)\setminus \outactions$
  \item    $\coactions = \coactions_1 \cup \coactions_2$
  \item    $\clocks = \clocks_1\cup \clocks_2$
  \item    $C_0 = C_0^1\cup C_0^2$
  \end{inparaenum}
  and $\trans$ is   defined by rules in
  Table~\ref{tb:parcomp} where we write $s\pll t$ instead of
  $(s,t)$.
\end{definition}

\begin{table}[t]
  \caption{Parallel composition on \iosac}\label{tb:parcomp}  \vspace{-1em}\
  \begin{minipage}{.48\textwidth}
  \begin{gather*}
    \hspace{-1em}
    \dedrule{s_1\trans[C,\act,C']_1s_1'}{s_1\pll s_2\trans[C,\act,C']s_1'\pll s_2}    \ \act{\in}(\actions_1{\!\setminus}\actions_2){\cup}\{\tau\}\! \tag{R1}\label{eq:par:l}
  \end{gather*}
  \end{minipage}
  \hfill
  \begin{minipage}{.48\textwidth}
  \begin{gather*}
    \hspace{-1em}
    \dedrule{s_2\trans[C,\act,C']_2s_2'}{s_1\pll s_2\trans[C,\act,C']s_1\pll s_2'}    \ \act{\in}(\actions_2{\!\setminus}\actions_1){\cup}\{\tau\}\! \tag{R2}\label{eq:par:r}
  \end{gather*}
  \end{minipage}

  ~\hfill
  \begin{minipage}{.68\textwidth}
  \begin{gather*}
    \dedrule{s_1\trans[C_1,\act,C'_1]_1s_1' \quad s_2\trans[C_2,\act,C'_2]_2s_2'}            {s_1\pll s_2\trans[C_1\cup C_2,\act,C'_1\cup C'_2]s_1'\pll s_2'}            \ \act{\in}(\actions_1{\cap}\actions_2){\setminus}\{\tau\}
            \tag{R3}\label{eq:par:s}
  \end{gather*}
  \end{minipage}
  \hfill~
\end{table}

Def~\ref{def:parcomp} does not ensure \emph{a priori} that the
resulting structure satisfies conditions
\ref{def:iosa:input-and-commit-are-reactive}--\ref{def:iosa:clock-never-go-off-early}
in Def.~\ref{def:iosac}.  This is only guaranteed by the following
proposition.

\begin{proposition}\label{prop:comp:is:closed}
  Let $\II_1$ and $\II_2$ be two compatible \iosa{s}. Then
  $\II_1\pll\II_2$ is indeed an \iosa.
\end{proposition}

\begin{figure}[t]
    \center
    \begin{tikzpicture}[shorten >=1pt,node distance=3cm,on grid,auto, scale=.8, transform shape, initial text={}]
        \tikzstyle{every state}=[rectangle, rounded corners=7pt, text=black,draw=black,minimum size=.5cm]
        \node[state,initial] at (0,0)     (s1)    {$s0\pll s3\pll s6$};
        \node[state]         at (3,0)     (s2)    {$s1\pll s3\pll s6$};
        \node[state]         at (6,0)     (s3)    {$s2\pll s3\pll s7$};
        \node[state]         at (0,-1.4)  (s4)    {$s0\pll s4\pll s6$};
        \node[state]         at (3,-1.4)  (s5)    {$s1\pll s4\pll s6$};
        \node[state]         at (9,-1.4)  (s6)    {$s2\pll s4\pll s7$};
        \node[state]         at (0,-2.8)  (s7)    {$s0\pll s5\pll s9$};
        \node[state]         at (3,-2.8)  (s8)    {$s1\pll s5\pll s9$};
        \node[state]         at (6,-2.8)  (s9)    {$s2\pll s5\pll s9$};
        \node[state]         at (9,-2.8)  (s10)   {$s2\pll s5\pll s8$};
        \path[->]
        (s1) edge []  node {$\{x\},a!,\emptyset$}       (s2)
        (s1) edge []  node {$\{y\},b!,\emptyset$}       (s4)
        (s2) edge []  node {$\emptyset,c!!,\emptyset$}  (s3)
        (s2) edge []  node {$\{y\},b!,\emptyset$}       (s5)
        (s3) edge []  node {$\{y\},b!,\emptyset$}       (s6)
        (s4) edge []  node {$\{x\},a!,\emptyset$}       (s5)
        (s4) edge []  node {$\emptyset,d!!,\emptyset$}  (s7)
        (s5) edge []  node {$\emptyset,c!!,\emptyset$}  (s6)
        (s5) edge []  node {$\emptyset,d!!,\emptyset$}  (s8)
        (s6) edge [left]  node {$\emptyset,d!!,\emptyset$}  (s10)
        (s7) edge []  node {$\{x\},a!,\emptyset$}       (s8)
        (s8) edge []  node {$\emptyset,c!!,\emptyset$}  (s9)
        (s10)edge []  node[above] {$\{x\},e!,\emptyset$}(s9)
        ;
    \end{tikzpicture}

\caption{\iosa\ resulting from the composition $\II_1\pll \II_2\pll \II_3$ of \iosa{s} in Fig.~\ref{fig:components}.}
\label{fig:composed}
\end{figure}

\begin{example}\label{ex:iosas:pll}
  The result of composing $\II_1\pll\II_2\pll\II_3$ from
  Example~\ref{ex:iosas} is depicted in Fig.~\ref{fig:composed}.
\end{example}

Larsen and Skou's probabilistic bisimulation~\cite{LarsenS91:iandc}
has been extended to NLMPs in~\cite{DArgenioSTW12:mscs}.  It can be
shown that the bisimulation equivalence is a congruence for parallel
composition of \iosa.  In fact, this has already been shown for
\iosa\ without urgency in~\cite{DArgenio2016} and since the
characteristics of urgency do not play any role in the proof over
there, the result immediately extends to our setting.  So we report
the theorem and invite the reader to read the proof
in~\cite{DArgenio2016}.

\begin{theorem}
  Let $\bisim$ denote the bisimulation equivalence relation on
  NLMPs~\cite{DArgenioSTW12:mscs} properly lifted to
  \iosa~\cite{DArgenio2016}, and let $\II_1$, $\II'_1$, $\II_2$,
  $\II'_2$ be \iosa{s} such that $\II_1\bisim \II'_1$ $\II_2\bisim
  \II'_2$.  Then, $\II_1\pll\II_2\bisim \II'_1\pll\II'_2$.
\end{theorem}

\section{Confluence}\label{sec:confluence}

\begin{wrapfigure}[5]{r}{120pt}
    \centering\vspace{-7em}
    \begin{tikzpicture}[shorten >=1pt,node distance=2.2cm,on grid,auto, scale=0.8, transform shape, initial text={}, initial where=above]
        \tikzstyle{every state}=[text=black,draw=black,fill=white,minimum size=.7cm]
        \draw [dashed] (-0.5,-2.7) -- (2.8,0.6);
        \node at (-0.8,0.2) {\LARGE $\forall$};
        \node at (3.1,-2.3) {\LARGE $\exists$};
        \node[state]    (s)                  {$s$};
        \node[state]    (s1)   [right=of s]  {$s_1$};
        \node[state]    (s2)   [below=of s]  {$s_2$};
        \node[state]    (s3)   [right=of s2] {$s_3$};
        \path[->]
        (s) edge []   node {$\emptyset,\act,C_1$}     (s1)
        (s) edge []   node[rotate=-90,yshift=.3cm,xshift=-.8cm]  {$\emptyset,\actb,C_2$}    (s2)
        (s1) edge []  node[rotate=-90,yshift=.3cm,xshift=-.8cm] {$\emptyset,\actb,C_2$}    (s3)
        (s2) edge []  node {$\emptyset,\act,C_1$}     (s3);
    \end{tikzpicture}
\caption{Confluence in \iosa.}\label{fig:confluence}
\end{wrapfigure}
Confluence, as studied by Milner~\cite{Milner1989:book}, is related
to a form of weak determinism: two silent transitions taking place on
an interleaving manner do not alter the behaviour of the process
regardless of which happens first.  In particular, we will eventually
assume that urgent actions in a closed \iosa\ are silent as they do
not delay the execution.  Thus we focus on confluence of urgent
actions only.
The notion of confluence is depicted in Fig.~\ref{fig:confluence} and
formally defined as follows.

\begin{definition}\label{def:confluence}
  An \iosa\ $\II$ is \emph{confluent with respect to actions
    $a,b\in\coactions$} if, for every state $s\in\states$ and
  transitions $s\trans[\emptyset,a,C_1]s_1$ and
  $s\trans[\emptyset,b,C_2]s_2$, there exists a state $s_3\in\states$
  such that $s_1\trans[\emptyset,b,C_2]s_3$ and
  $s_2\trans[\emptyset,a,C_1]s_3$.
  $\II$ is \emph{confluent} if it is confluent with respect to every
  pair of urgent actions.
              \end{definition}
Note that we are asking that the two actions converge in a single
state, which is stronger than Milner's strong confluence, where
convergence takes place on bisimilar but potentially different states.

Confluence is preserved by parallel composition: 
\begin{proposition}  \label{prop:compPreservesConfl}
  If both $\II_1$ and $\II_2$ are confluent w.r.t.\ actions
  $a,b\in\coactions$, then so is $\II_1\pll\II_2$.
  Therefore, if $\II_1$ and $\II_2$ are confluent,
  $\II_1\pll\II_2$ is also confluent.
\end{proposition}

However, parallel composition may turn non-confluent components into a
confluent composed system.

By looking at the \iosa\ in Fig.~\ref{fig:confluence-is-det}, one can
notice that the non-determinism introduced by confluent urgent output
actions is spurious in the sense that it does not change the
stochastic behaviour of the model after the output urgent actions have been abstracted.  Indeed, since time does not
progress, it is the same to sample first clock $x$ and then clock $y$
passing through state $s_1$, or first $y$ and then $x$ passing through
$s_2$, or even sampling both clocks simultaneously through a transition
$s_1\trans[\emptyset,\tau,\{x,y\}]s_3$.  In any of the cases, the
stochastic resolution of the execution of $\act$ or $\actb$ in the
stable state $s_3$ is the same.
This could be generalized to any number of confluent transitions.

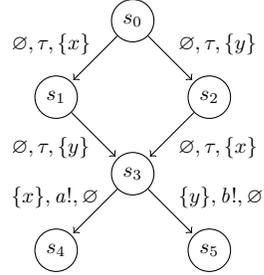
\begin{wrapfigure}[12]{r}{3.4cm}
    \centering\vspace{-2.6em}
    \begin{tikzpicture}[shorten >=1pt,node distance=1.8cm,on grid,auto, scale=0.8, transform shape, initial text={}, initial where=above]
        \tikzstyle{every state}=[text=black,draw=black,fill=white,minimum size=.7cm]
        \node[state] (s0)                       {$s_0$};
        \node[state] (s1) [below left = of s0]  {$s_1$};
        \node[state] (s2) [below right = of s0] {$s_2$};
        \node[state] (s3) [below right = of s1] {$s_3$};
        \node[state] (s4) [below left = of s3]  {$s_4$};
        \node[state] (s5) [below right = of s3] {$s_5$};
        \path[->]
        (s0) edge []   node[yshift=.55cm,xshift=-1.5cm] {$\emptyset,\tau,\{x\}$}     (s1)
        (s0) edge []   node[yshift=-.05cm,xshift=0cm]  {$\emptyset,\tau,\{y\}$}    (s2)
        (s1) edge []  node[yshift=-.5cm,xshift=-1.5cm] {$\emptyset,\tau,\{y\}$}    (s3)
        (s2) edge []  node[yshift=.1cm,xshift=0cm] {$\emptyset,\tau,\{x\}$}     (s3)
        (s3) edge []   node[yshift=.55cm,xshift=-1.5cm] {$\{x\},\act!,\emptyset$}     (s4)
        (s3) edge []   node[yshift=-0.05cm,xshift=0cm]  {$\{y\},\actb!,\emptyset$}    (s5);
    \end{tikzpicture}
\caption{Confluence is weakly deterministic}\label{fig:confluence-is-det}
\end{wrapfigure}

Thus, it will be convenient to use term rewriting techniques
to collect all clocks that are active in the convergent stable state
and have been activated through a path of urgent actions.
Therefore, we recall some basic notions of rewriting systems.
An \emph{abstract reduction system}~\cite{DBLP:books/daglib/0092409}
is a pair $(\E,\reduct)$, where the reduction $\reduct$ is a binary
relation over the set $\E$, i.e. ${\reduct}\subseteq\E\times\E$.  We
write $a\reduct b$ for $(a,b)\in{\reduct}$.
We also write $a\reduct[*]b$ to denote that there is a path $a_0\reduct
a_1\dots\reduct a_n$ with $n\geq 0$, $a_0=a$ and $a_n=b$.  An element
$a\in\E$ is in \emph{normal form} if there is no $b$ such that
$a\reduct b$.  We say that $b$ is a normal form of $a$ if $a\reduct[*]b$
and $b$ is in normal form.
A reduction system $(\E,\reduct)$ is \emph{confluent} if for all
$a,b,c\in\E$ $a\invreduct[*]c\reduct[*]b$ implies
$a\reduct[*]d\invreduct[*]b$ for some $d\in\E$.  This notion of
confluence is implied by the following statement: for
all $a,b,c\in\E$, $a\invreduct[]c\reduct[] b$ implies that either
$a\reduct[]d\invreduct[]b$ for some $d\in\E$, or $a=b$.
A reduction system is \emph{normalizing} if every element has a normal
form, and it is \emph{terminating} if there is no infinite chain
$a_0\reduct a_1\reduct\cdots$.  A terminating reduction system is also
normalizing.
In a confluent reduction system every element has at most one normal
form. If in addition it is also normalizing, then the normal form is
unique.

We now define the abstract reduction system introduced by the urgent
transitions of an \iosa.

\begin{definition}\label{def:urgentabstractreduction}
  Given an \iosa\ $\I=(\states,\actions,\clocks,\trans_{\I},C_0,s_0)$,
  define the abstract reduction system $\uars_I$ as
  $(\states\times\partsof(\clocks)\times\N_0,\reduct)$ where
  $(s,C,n)\reduct (s',C\cup C',n+1)$ if and only if there exists
  $a\in\coactions$ such that $s\trans[\emptyset,a,C'] s'$.
\end{definition}

An \iosa\ is \emph{non-Zeno} if there is no loop of urgent actions.
The following result can be straightforwardly proven.

\begin{proposition}\label{prop:uarnormandconfl}
  Let the \iosa\ $\I$ be closed and confluent.  Then $\uars_\I$ is
  confluent, and hence every element has at most one normal form.
  Moreover, an element $(s,C,n)$ is in normal form iff $s$ is stable
  in $\I$.
  If in addition $\I$ is non-Zeno, $\uars_\I$ is also terminating and
  hence every element has a unique normal form.
        \end{proposition}

\section{Weak determinism}\label{sec:weak:det}

As already shown in Fig.~\ref{fig:confluence-is-det}, the
non-determinism introduced by confluence is spurious.
In this section, we show that closed confluent \iosa{s} behave
deterministically in the sense that the stochastic behaviour of the
model is the same, regardless the way in which non-determinism is
resolved.
Thus, we say that a closed \iosa\ is \emph{weakly deterministic} if
\begin{inparaenum}[(i)]
\item  almost surely at most one discrete non-urgent transition is enabled
  at every time point,
\item  the election over enabled urgent transitions does not affect the non
  urgent-behavior of the model, and
\item  no non-urgent output and urgent output are enabled simultaneously.
\end{inparaenum}
To avoid referring explicitly to time in (i), we say instead that a closed
\iosa\ is weakly deterministic if it almost never reaches a state in which
two different non-urgent discrete transitions are enabled.  Moreover,
to ensure (ii), we define the following weak transition.

For this definition and the rest of the section we will assume that
the \iosa\ is closed and all its urgent actions have been abstracted,
that is, all actions in $\coactions$ have been renamed to $\tau$.

\begin{definition}\label{def:weaktransition}
  For a non stable state $s$, and $v\in\R^N$, we define
  $(s,\vec{v})\Trans[C]_n\mu$ inductively by the following rules:
  \begin{gather*}
    \text{\emph{(T1)}}\quad
    \dedrule{\begin{array}{c}s\trans[\emptyset,\tau,C]s' \\ \stable{s'}\end{array}}
            {(s,\vec{v})\Trans[C]_1\mu^{\vec{v}}_{C,s'}}
    \qquad\quad
    \text{\emph{(T2)}}\quad
    \dedrule{\begin{array}{c}s\trans[\emptyset,\tau,C']s' \\
             \forall {\vec{v}'\in\R^N}:\exists {C'',\mu'}: (s',\vec{v}')\Trans[C'']_n\mu'\end{array}}
            {(s,\vec{v})\Trans[C'\cup C'']_{n+1}\hat{\mu}}  \end{gather*}
  where
  $\mu^{\vec{v}}_{C,s}$ is defined as in Def.~\ref{def:iosasemantic} and
  $\hat{\mu}=\int_{\states\times\R^N} f_n^{C''}d\mu^{\vec{v}}_{C',s'}$, with
  $f_n^{C''}(t,\vec{w})= \nu$, if $(t,\vec{w})\Trans[C'']_n\nu$, and
  $f_n^{C''}(t,\vec{w})= \emptyfun$ otherwise.
  We define the \emph{weak transition} $(s,\vec{v})\Trans[]\mu$ if
  $(s,\vec{v})\Trans[C]_n\mu$ for some $n\geq 1$ and
  $C\subseteq\clocks$.
\end{definition}

As given above, there is no guarantee that $\Trans[C]_n$ is well
defined.  In particular, there is no guarantee that $f_n^{C''}$ is a
well defined measurable function.
We postpone this to Lemma~\ref{lm:weak:transition} below.

With this definition, we can introduce the concept of weak
determinism:

\begin{definition}\label{def:deterministiciosa}
  A closed \iosa\ $\I$ is \emph{weakly deterministic} if $\Trans[]$ is
  well defined in $\I$ and, in $P(\I)$,
  any state $(s,v) \in\pstates$ that satisfies one of the following
  conditions is almost never reached from any
  $(\init,v_0)\in\pstates$:
  \begin{inparaenum}[(a)]
  \item\label{def:determfirst}    $s$ is stable and $\cup_{a\in\actions\cup\{\init\}}\transitions_a
    (s,v)$ contains at least two different probability measures,
  \item\label{def:determsecond}    $s$ is not stable, $(s,v)\Trans[]\mu$, $(s,v)\Trans[]\mu'$ and
    $\mu\neq\mu'$, or
  \item\label{def:determthird}    $s$ is not stable and $(s,v)\trans[a]\mu$ for some
    $a\in\outactions\setminus\coactions$.
  \end{inparaenum}
\end{definition}

By ``almost never'' we mean that the measure of the set of all paths
leading to any measurable set in $\borel(\pstates)$ containing only
states satisfying (a), (b), or (c) is zero.
Thus, Def.~\ref{def:deterministiciosa} states that, in a weakly
deterministic \iosa, a situation in which a non urgent output action
is enabled with another output action, being it urgent (case (c)) or
non urgent (case (a)), or in which sequences of urgent transitions lead to
different stable situations (case (b)), is almost never reached.

For the previous definition to make sense we need that $\P(\I)$
satisfies \emph{time additivity}, \emph{time determinism}, and
\emph{maximal progress}~\cite{Yi90:concur}.  This is stated in the
following theorem whose proof follows as
in~\cite[Theorem~16]{DArgenio2016}.
\begin{theorem}
  Let $\I$ be an \iosa\ $\I$.  Its semantics $\P(\I)$
  satisfies, for all $(s,\vec{v})\in\pstates$, $\act\in\outactions$ and
  $d,d'\in\R_{>0}$,
  \begin{inparaenum}[(i)]
  \item    $\transitions_\act(s,\vec{v})\neq\emptyset \ \Rightarrow \
    \transitions_d(s,\vec{v})=\emptyset$
    (maximal progress),
  \item    $\mu,\mu'\in\transitions_d(s,\vec{v}) \ \Rightarrow \ \mu=\mu'$
    (time determinism), and
  \item    ${\delta^{-d}_{(s,\vec{v})}{\in}\transitions_d(s,\vec{v})}
    \wedge
    {\delta^{-d'}_{(s,\vec{v}-d)}{\in}\transitions_{d'}(s,\vec{v}-d)}
    \ \iff \
    {\delta^{-(d+d')}_{(s,\vec{v})}{\in}\transitions_{d+d'}(s,\vec{v})}$
    (time additivity).
  \end{inparaenum}
\end{theorem}

The next lemma states that, under the hypothesis that the
\iosa\ is closed and confluent, $\Trans[C]_n$ is well defined.
Simultaneously, we prove that $\Trans[C]_n$ is deterministic.

\begin{lemma}\label{lm:weak:transition}
  Let $\I$ be a closed and confluent \iosa.  Then, for all $n\geq 1$,
  the following holds:
  \begin{enumerate}
  \item    If $(s,\vec{v})\Trans[C]_n\mu$ then there is a stable state $s'$
    such that
    \begin{inparaenum}[(i)]
    \item      $\mu=\mu^{\vec{v}}_{C,s'}$,
    \item      $(s,C',m)\reduct[*](s',C'{\cup}C,m{+}n)$ for all
      $C'\subseteq\clocks$ and $m\geq0$, and
    \item      if $(s,\vec{v}')\Trans[C']_n\mu'$ then $C'=C$ and moreover, if
      $\vec{v}'=\vec{v}$, also $\mu'=\mu$; and
    \end{inparaenum}
  \item    $f^C_n$ is a measurable function.
  \end{enumerate}
\end{lemma}

The proof of the preceding lemma uses induction on $n$ to prove item 1
and 2 simultaneously.  It makes use of the previous results on
rewriting systems in conjunction with measure theoretical tools such
as Fubini's theorem to deal with Lebesgue integrals on product spaces.
All these tools make the proof that confluence preserves weak
determinism radically different from those of
Milner~\cite{Milner1989:book} and Crouzen~\cite{Crouzen2014:phd}.

The following corollary follows by items 1.(ii) and 1.(iii) of Lemma~\ref{lm:weak:transition}. 
\begin{corollary}\label{cor:weak:transition}
  Let $\I$ be a closed and confluent \iosa.  Then, for all
  $(s,\vec{v})$, if $(s,\vec{v})\Trans[]\mu_1$ and
  $(s,\vec{v})\Trans[]\mu_2$, $\mu_1=\mu_2$.
\end{corollary}

This corollary already shows that closed and confluent \iosa{s}
satisfy part (b) of Def.~\ref{def:deterministiciosa}.  In general, we
can state:

\begin{theorem}\label{th:iosacdeterministic}
  Every closed confluent \iosa\ is weakly deterministic.
\end{theorem}

The rest of the section is devoted to discuss the proof of this theorem.  From now
on, we work with the closed confluent
\iosa\ $\I=(\istates,\clocks,\actions,\trans,s_0,C_0)$, with
$|\clocks|=N$, and its semantics
$\P(\I)=(\pstates,\borel(\pstates),\{\transitions_\act\mid\act\in\labels\})$.

The idea of the proof of Theorem~\ref{th:iosacdeterministic} is to
show that the property that all active clocks have non-negative values
and they are different from each other is almost surely an invariant
of $\I$, and that at most one non-urgent transition is enabled in
every state satisfying such invariant. Furthermore, we want to show
that, for unstable states, active clocks have strictly positive
values, which implies that non-urgent transitions are never enabled in
these states.
Formally, the invariant is the set
\begin{align}
  \inv =
  & \phantom{{}\cup{}}
         \{(s,\vec{v})\mid \stable{s} \text{ and }
           \forall x_i,x_j \in\activeck(s):
           i\neq j \Rightarrow \vec{v}(i)\neq\vec{v}(j)
           \land \vec{v}(i) \geq 0\}\notag\\
  & \cup \{(s,\vec{v})\mid \neg\stable{s} \text{ and }
           \forall x_i,x_j \in\activeck(s):
           i\neq j \Rightarrow \vec{v}(i)\neq\vec{v}(j)
           \vec{v}(i)>0\}\notag\\
  & \cup (\{\init\}\times\R^N)
\end{align}
with $\activeck$ as in Def.~\ref{def:iosac}.
Note that its complement is:
\begin{align}
  \notinv =
  & \phantom{{}\cup{}}
         \{(s,\vec{v})\mid            \exists x_i,x_j \in\activeck(s):
           i \neq j \land \vec{v}(i)=\vec{v}(j)\}\notag \\
  & \cup \{(s,\vec{v})\mid \stable{s} \text{ and }
           \exists x_i \in\activeck(s): \vec{v}(i)<0 \}\notag\\
  & \cup \{(s,\vec{v})\mid \neg\stable{s}\text{ and }
           \exists x_i \in\activeck(s):\vec{v}(i)\leq0\}
\end{align}

It is not difficult to show that $\notinv$ is measurable and, in
consequence, so is $\inv$.
The following lemma states that $\notinv$ is almost never reached in
one step from a state satisfying the invariant.

\begin{lemma}\label{lemma:one-step-invariant}
  If $(s,\vec{v})\in\inv$, $a\in\labels$, and
  $\mu\in\transitions_a(s,\vec{v})$, then $\mu(\notinv)=0$.
\end{lemma}

From this lemma we have the following corollary

\begin{corollary}\label{cor:invariant}
  The set $\notinv$ is almost never reachable in $\P(\I)$.
\end{corollary}

The proof of the corollary requires the definitions related to
schedulers and measures on paths in NLMPs (see
\cite[Chap.~7]{Wolovick2012:phd} for a formal definition of scheduler
and probability measures on paths in NLMPs.)
We omit the proof of the corollary since it eventually boils down to
an inductive application of Lemma~\ref{lemma:one-step-invariant}.

The next lemma states that any stable state in the invariant $\inv$
has at most one discrete transition enabled. Its proof is the same as
that of \cite[Lemma~20]{DArgenio2016}.

\begin{lemma}\label{lemma:inv-enables-at-most-one}
  For all $(s,\vec{v})\in\inv$ with $s$ stable or $s=\init$, the set
  $\bigcup_{\act\in\actions\cup\{\init\}}\transitions_\act(s,\vec{v})$
  is either a singleton set or the empty set.
\end{lemma}

The next lemma states that any unstable state in the invariant $\inv$
can only produce urgent actions.

\begin{lemma}
\label{lem:nonurgentunreachable}
  For every state $(s,\vec{v})\in\inv$, if $\neg\stable{s}$
  and $(s,\vec{v})\trans[a]\mu$, then $a\in\coactions$.
\end{lemma}
\begin{proof}
  First recall that $\I$ is closed; hence $\inactions=\emptyset$.
  If $(s,\vec{v})\in\inv$ and $\neg\stable{s}$ then $\vec{v}_i>0$ for
  all $x_i\in\enablingck(s)\subseteq\activeck(s)$. Therefore, by
  Def.~\ref{def:iosasemantic},
  $\transitions_\act(s,\vec{v})=\emptyset$ if
  $\act\in\outactions\setminus\coactions$.
  Furthermore, for any $d\in\R_{>0}$,
  $\transitions_d(s,\vec{v})=\emptyset$ since $s$ is not stable and
  hence $s\trans[\_,\actb,\_]\_$ for some
  $\actb\in\outactions\cap\coactions$.
  \qed
\end{proof}

Finally, Theorem~\ref{th:iosacdeterministic} is a consequence of
Lemma~\ref{lemma:inv-enables-at-most-one},
Lemma~\ref{lem:nonurgentunreachable}, Corollary~\ref{cor:invariant},
and Corollary~\ref{cor:weak:transition}.

\section{Sufficient conditions for weak determinism}\label{sec:ensuring-confluence}

Fig.~\ref{fig:composed} shows an example in which the composed
\iosa\ is weakly deterministic despite that some of its components are not
confluent.
The potential non-determinism introduced in state $s_1\pll s_4\pll s_6$
is never reached since urgent actions at states $s_0\pll s_4\pll s_6$
and $s_1\pll s_3\pll s_6$ prevent the execution of non urgent actions
leading to such state.
We say that state $s_1\pll s_4\pll s_6$ is not \emph{potentially
  reachable}.  The concept of potentially reachable can be
defined as follows.

\begin{definition}\label{def:pot:reach}
  Given an \iosa\ $\I$, a state $s$ is \emph{potentially reachable} if
  there is a path
  $s_0\trans[\_,a_0,\_]s_1\ldots,s_{n-1}\trans[\_,a_{n-1},\_]s_n=s$
  from the initial state, with $n\geq 0$, such that for all $0\leq i<n$, if
  $s_i\trans[\_,b,\_]\_$ for some $b\in\coactions\cap\outactions$
  then $a_i\in\coactions$.  In such case we call the path
  \emph{plausible}.
\end{definition}

Notice that none of the paths leading to  $s_1\pll s_4\pll s_6$ in
Fig.~\ref{fig:composed} are plausible.
Also, notice that an \iosa\ is bisimilar to the same \iosa\ when its
set of states is restricted to only potentially reachable states.

\begin{proposition}\label{prop:only:pot:reach}
  Let $\I$ be a closed \iosa\ with set of states $\states$ and let
  $\overline{\I}$ be the same \iosa\ as $\I$ restricted to the set of
  states
  $\overline{\states}=\{s\in\states\mid \text{ is potentially reachable in }\I\}$.
  Then $\I\bisim\overline{\I}$.
    \end{proposition}

Although we have not formally introduced bisimulation, it should
be clear that both semantics are bisimilar through the identity
relation since a transition $s\trans[\{x\},a,C]s'$ with $s$ unstable does
not introduce any concrete transition. (Recall the \iosa\ is
closed so there is no input action on $\I$.)

For a state in a composed \iosa\ to be potentially reachable,
necessarily each of the component states has to be potentially
reachable in its respective component \iosa.

\begin{lemma}
\label{lemma:streach}
  If a state $s_1\pll\cdots\pll s_n$ is potentially reachable in
  $\II_1\pll\cdots\pll \II_n$ then $s_i$ is potentially
  reachable in $\II_i$ for all $i=1,\dots,N$.
\end{lemma}

By Theorem~\ref{th:iosacdeterministic}, it suffices to check whether
a closed \iosa\ is confluent to ensure that it is weakly deterministic.
In this section, and following ideas introduced
in~\cite{Crouzen2014:phd}, we build on a theory that allows us to
ensure that a closed composed \iosa\ is confluent in a compositional
manner, even when its components may not be confluent.
Theorem~\ref{theo:sufCondForDet} provides the sufficient conditions to
guarantee that the composed \iosa\ is confluent.
Because of Proposition~\ref{prop:compPreservesConfl}, it suffices to
check whether two urgent actions that are not confluent in a single
component are potentially reached.  Since potential reachability
depends on the composition, the idea is to overapproximate by
inspecting the components.  The rest of the section builds on concepts
that are essential to construct such overapproximation.

Let $\uen(s)=\{a\in\coactions\mid s\trans[\_,a,\_]\_\}$ be the set of
urgent actions enabled in a state $s$. We say that a set $B$ of output
urgent actions is spontaneously enabled by a non-urgent action $b$ if
$b$ is potentially reached and it transitions to a state enabling all
actions in $B$.

\begin{definition}  A set $B\subseteq\coactions\cap\outactions$ is \emph{spontaneously
    enabled} by $\act\in\actions\setminus\coactions$ in $\II$, if either
  $B=\emptyset$ or there are potentially reachable states $s$ and $s'$
  such that $s$ is stable, $s\trans[\_,\act,\_]s'$, and
  $B\subseteq\uen(s')$.
  $B$ is \emph{maximal} if for any $B'$ spontaneously enabled by $b$
  in $\II$ such that $B\subseteq B'$, $B=B'$.
                    \end{definition}

A set that is spontaneously enabled in a composed \iosa, can be
constructed as the union of spontaneously enabled sets in each of the
components as stated by the following proposition.  Therefore,
spontaneously enabled sets in a composed \iosa\ can be overapproximated
by unions of spontaneously enabled sets of its components.

\begin{proposition}\label{prop:overaprox:sp:enab}
  Let $B$ be spontaneously enabled by action $a$ in
  $\II_1\pll\ldots\pll\II_n$.
  Then, there are $B_1,\ldots,B_n$ such that each $B_i$ is
  spontaneously enabled by $a$ in $\II_i$, and $B=\bigcup_{i=1}^nB_i$.
  If in addition $B$ is maximal, there are $B_1,\ldots,B_n$ such that
  each $B_i$ is maximal spontaneously enabled by $a$ in $\II_i$,
  and $B\subseteq\bigcup_{i=1}^nB_i$.
      \end{proposition}
\begin{proof}
  We only prove it for $\II_1\pll\II_2$. The generalization to any $n$
  follows easily.
  Let $\bar{B}_i=B\cap\actions_i$ for $i=1,2$ and note that
  $B=\bar{B}_1\cup\bar{B}_2$.  We show that $\bar{B}_1$ is
  spontaneously enabled by $a$ in $\II_1$.  The case of $\bar{B}_2$
  follows similarly.
  Since $B$ is spontaneously enabled by $a$ in $\II_1\pll\II_2$,
  there exist potentially reachable states $s_1\pll s_2$ and
  $s'_1\pll s'_2$, such that $s_1\pll s_2$ is stable,
  $s_1\pll s_2\trans[\_,a,\_]s'_1\pll s'_2$, and
  $B\subseteq\uen(s'_1\pll s'_2)$.
  First notice that $\bar{B}_1\subseteq\uen(s_1)$.  Also, suppose
  $\bar{B}_1\neq\emptyset$, otherwise $\bar{B}_1$ is spontaneously
  enabled by $a$ trivially.
  Consider first the case that $a\in \actions_2\setminus\actions_1$.
  By $\textrm{(R2)}$, $s_1=s'_1$, but, since there is some
  $b\in\bar{B}_1$, $s_1\trans[\_,b,\_]\_$ and hence
  $s_1\pll s_2\trans[\_,b,\_]\_$ rendering $s_1\pll s_2$ unstable, which
  is a contradiction.
  So $a\in\actions_1$ and $s_1\trans[\_,a,\_]s'_1$. By
  Lemma~\ref{lemma:streach}, $s_1$ and $s'_1$ are potentially
  reachable and, necessarily, $s_1$ is stable (otherwise $s_1\pll s_2$
  has to be unstable as shown before).  Therefore $\bar{B}_1$ is
  spontaneously enabled by $a$ in $\II_1$.
  The second part of the proposition is immediate from the first part.
                                        \qed
\end{proof}

Spontaneously enabled sets refer to sets of urgent output actions that
are enabled after some steps of execution.  Urgent output actions can
also be enabled at the initial state.

\begin{definition}  \label{def:init:set}
  A set $B\subseteq\coactions\cap\outactions$ is \emph{initial} in an
  \iosa\ $\II$ if $B\subseteq\uen(s_0)$, with $s_0$ being the initial
  state of $\II$.
  $B$ is \emph{maximal} if $B=\uen(s_0)\cap\outactions$.
                  \end{definition}

An initial set of a composed \iosa\ can be constructed as the union of
initial sets of its components.  In particular the maximal initial set
is the union of all the maximal sets of its components.  The proof
follows directly from the definition of parallel composition taking
into consideration that \iosa{s} are input enabled.

\begin{proposition}\label{prop:initAdit}
  Let $B$ be initial in $\II=(\II_1 \pll \ldots$ $\pll \II_n)$.  Then,
  there are $B_1,\ldots,B_2$, with $B_i$ initial of $\II_i$, $1\leq
  i\leq n$ and $B=\bigcup_{i=1}^nB_i$.
  Moreover,
  $\uen(s_0)\cap\outactions_\II=\bigcup_{i=1}^n\uen(s^0_i)\cap\outactions_i$.
        \end{proposition}

We say that an urgent action triggers an urgent output action if the
first one enables the occurrence of the second one, which was not
enabled before.

\begin{definition}  \label{def:trigrel}
  Let $a \in \coactions$ and $b \in \coactions\cap\outactions$.
  $a$ \emph{triggers} $b$ in an \iosa\ $\II$ if there are potentially
  reachable states $s_1$, $s_2,$ and $s_3$ such that
  $s_1\trans[\_,a,\_]s_2\trans[\_,b,\_]s_3$ and, if $a\neq b$,
  $b\notin\uen(s_1)$.
                    
              \end{definition}

Notice that, for the particular case in which $a=b$, $b\notin\uen(s)$
is not required.
The following proposition states that if one action triggers another
one in a composed \iosa, then the same triggering occurs in a particular
component.

\begin{proposition}\label{prop:triggApprox}
  Let $a \in \coactions$ and $b \in \coactions\cap\outactions$ such
  that $a$ triggers $b$ in $\II_1\pll\ldots\pll \II_n$.  Then there is
  a component $\II_i$ such that $b\in \outactions_i$ and $a$ triggers
  $b$ in $\II_i$.
          \end{proposition}
\begin{proof}
  We only prove it for $\II_1\pll\II_2$. The generalization to any $n$
  follows easily.
  Because $b\in\coactions\cap\outactions$ necessarily
  $b\in\outactions_1$ or $b\in\outactions_2$.  W.l.o.g.\ suppose
  $b\in\outactions_1$.
  Since $a$ triggers $b$ in $\II_1\pll\II_1$,
  $s_1\pll s_2\trans[\_,a,\_]s'_1\pll s'_2\trans[\_,b,\_]s''_1\pll s''_2$
  with $s_1\pll s_2$, $s'_1\pll s'_2$, and $s''_1\pll s''_2$ being
  potentially reachable.

  Suppose first that $a\neq b$. Then $b\notin\uen(s_1\pll s_2)$.
  Recall that, by Lemma~\ref{lemma:streach}, $s_1$, $s'_1$, and
  $s''_1$ are potentially reachable in $\II_1$.
  Since $b\in\outactions_1$, $s'_1\trans[\_,b,\_]s''_1$.  Suppose
  $a\in\actions_2 \setminus\actions_1$.  Then, necessarily, $s_1=s'_1$
  which gives $b\in\uen(s_1)\cap\outactions\subseteq\uen(s_1\pll s_2)$, yielding a
  contradiction.  Thus, necessarily $a\in\coactions_1$ and hence
  $s_1\trans[\_,a,\_]s'_1$, by the definition of parallel composition.
  It remains to show that $b\notin\uen(s_1)$, but this is immediate
  since $\uen(s_1)\cap\outactions\subseteq\uen(s_1\pll s_2)$ and
  $b\notin\uen(s_1\pll s_2)$.
  Thus $a$ triggers $b$ in $\II_1$ in this case.
  If instead $a=b$, by the definition of parallel composition we
  immediately have that $s_1\trans[\_,b,\_]s'_1\trans[\_,b,\_]s''_1$,
  proving thus the proposition.
                                            \qed
\end{proof}

Proposition~\ref{prop:triggApprox} tells us that the triggering
relation of a composed \iosa\ can be overapproximated by the union of
the triggering relations of its components.
Thus we define:

\begin{definition}
  The \emph{approximate triggering relation} of $\II_1 \pll \ldots
  \pll \II_n$ is defined by
  ${\triggers} = \bigcup_{i=1}^n \{ (a,b) \mid a \text{ triggers } b \text{ in } \II_i\}$.
  Its reflexive transitive closure $\triggers^*$ is called
  \emph{approximate indirect triggering relation}.
        \end{definition}

The next definition characterizes all sets of urgent output actions
that are simultaneously enabled in any potentially reachable state of a
given \iosa.

\begin{definition}
  A set $B\subseteq\coactions\cap\outactions$ is an \emph{enabled set}
  in an \iosa\ $\II$ if there is a potentially reachable state $s$
  such that $B\subseteq\uen(s)$.  If $a\in B$, we say that $a$ is
  \emph{enabled} in $s$.
  Let $\ES_\II$ be the set of all enabled sets in $\II$.
          \end{definition}

If an urgent output action is enabled in a potentially reachable state
of a \iosa, then it is either initial, spontaneously enabled, or
triggered by some action.

\begin{theorem}\label{theo:enabledTheo}
  Let $b\in\coactions\cap\outactions$ be enabled in some potentially
  reachable state of the \iosa\ $\II$.  Then there is a set $B$ with
  $b\in B$ that is either initial or spontaneously enabled by some
  action $a\in\coactions$, or $b$ is triggered by some action
  $a\in\outactions\setminus\coactions$.
        \end{theorem}
\begin{proof}
  Let $s$ be potentially reachable in $\II$ such that
  $b\in\uen(s)\cap\outactions$.
  We prove the theorem for $b$ by induction on the plausible path
  $\sigma$ leading to $s$.
  If $|\sigma|=0$, then $\sigma=s$ and $s$ is the initial state.  Then
  the set $\uen(s)\cap\outactions$ is initial and we are done in this
  case.
  If $|\sigma|>0$, then $\sigma=\sigma'\cdot(s'\trans[\_,a,\_]s)$ for
  some $s'$, $a$, and plausible $\sigma'$.
  If $a\in\actions\setminus\coactions$ then $s'$ is stable (since
  $\sigma$ is plausible) and thus $\uen(s)\cap\outactions$ is
  spontaneously enabled by $a$.
  If instead $a\in\coactions$, two possibilities arise. If
  $b\notin\uen(s')$, then $b$ is triggered by $a$.
  If $b\in\uen(s')$, the conditions are satisfied by induction since
  $|\sigma'|=|\sigma|-1$.
                                                      \qed
\end{proof}

The next definition is auxiliary to prove the main theorem of this
section.  It constructs a graph from a closed and composed
\iosa\ whose vertices are sets of urgent output actions.  It has the
property that, if there is a path from one vertex to another, all actions
in the second vertex are approximately indirectly triggered by actions
in the first vertex (Lemma~\ref{lemm:IndTrig}).  This will allow to
show that any set of simultaneously enabled urgent output actions is
approximately indirectly triggered by initial actions or spontaneously
enabled sets (Lemma~\ref{lm:ApproxCommonCause}).

\begin{definition}\label{def:egraph}
  Let $\II = (\II_1\pll\ldots\pll\II_n)$ be a closed \iosa.  The
  \emph{enabled graph} of $\II$ is defined by the labelled graph
  $\EG_\II=(V,E)$, where $V \subseteq 2^{\outactions\cap\coactions}$
  and $E\subseteq {V\times(\coactions{\cap}\outactions)\times V}$, with
  $V=\bigcup_{k\geq 0} V_k$ and $E=\bigcup_{k\geq 0} E_k$, and, for
  all $k\in\N$, $V_k$ and $E_k$ are inductively defined by
  \begin{align*}
    V_0 =\
    & \textstyle
      \bigcup_{a\in \actions}
        \{\bigcup_{i=1}^{n} B_i \mid
          \forall {1\leq i \leq n} :\\
    & \phantom{\textstyle\bigcup_{a\in \actions}\{\bigcup_{i=1}^{n}B_i\mid\ }
          B_i \text{ \em is spontaneously enabled by } a
          \text{ \em and maximal in } \II_i\}\\
    & \cup\textstyle
      \{\bigcup_{i=1}^{n} \uen(s^0_i)\cap\outactions_i \mid
      \forall {1\leq i\leq n} : s^0_i \text{ \em is the initial state in } \II_i\}\\
                E_k =\
    & \{(v,a,(v{\setminus}\{a\})\cup\{b \mid a{\triggers}b\}) \mid
        v\in V_k, a\in v \}\\
    V_{k+1} =\
    & \textstyle
      \{ v' \mid v\in V_i, (v,v')\in E_k, v'\notin\bigcup_{j=0}^{k}V_j\}
  \end{align*}
                                                  \end{definition}

Notice that $V_0$ contains the maximal initial set of $\I$ and an
overapproximation of all its maximal spontaneously enabled sets.
Notice also that, by construction, there is a path from any vertex in
$V$ to some vertex in $V_0$.

The set closure of $V$ in $\EG_\II$, defined by
$\overline{\ES}_\II = \{B \mid B\subseteq v, v\in V\}$,
turns out to be an overapproximation of the actual set $\ES_\II$ of
all enabled sets in $\II$.

\begin{lemma}\label{lm:enabledGraphSuper}
  For any closed \iosa\ $\II=(\II_1\pll\cdots\pll\II_n)$, $\ES_\II
  \subseteq \overline{\ES}_\II$.
        \end{lemma}
\begin{proof}
  Let $B\in\ES_\II$.  We proceed by induction on the length of the plausible path
  $\sigma$ that leads to the state $s$ s.t.\ $B\subseteq\uen(s)$.
  If $|\sigma|=0$ then $s$ is the initial state and thus $B$
  is initial in $\II$. Thus, by Def.~\ref{def:init:set},
  Prop.~\ref{prop:initAdit}, and Def.~\ref{def:egraph},
  $B \subseteq (\uen(s_0){\cap}\outactions_\II) =
  (\bigcup_{i=1}^n\uen(s^0_i){\cap}\outactions_i) \in V_0 \subseteq
  \overline{\ES}_\II$.
  As a consequence $B\in\overline{\ES}_\II$.

  If $|\sigma|>0$ then $\sigma=\sigma'\cdot(s'\trans[\_,a,\_]s)$, for
  some $s'$, $a$, and plausible $\sigma'$.
  If $a\in\actions\setminus\coactions$ then $s'$ is stable (since
  $\sigma$ is plausible) and thus $B$ is spontaneously enabled by $a$.
  By Prop.~\ref{prop:overaprox:sp:enab}, there are $B_1,\ldots,B_n$
  such that each $B_i$ is spontaneously enabled by $a$ and maximal in
  $\II_i$, and $B\subseteq\bigcup_{i=1}^nB_i$.  Since
  $\bigcup_{i=1}^nB_i\in V_0\subseteq\overline{\ES}_\II$, then
  $B\in\overline{\ES}_\II$.
  If instead $a\in\coactions$, let $B'=\{a\}\cup(B\cap\uen(s'))$.
  Notice that $B'\subseteq\uen(s')\cap\outactions$.  Since $s'$ is the
  last state on $\sigma'$ and $|\sigma'|=|\sigma|-1$,
  $B'\in\overline{\ES}_\II$ by induction.
  Hence, there is a vertex $v'\in V$ in $\EG_\II$ such that
  $B'\subseteq v$ and, by Def.~\ref{def:egraph}, $v'\in V_k$ for some
  $k\geq 0$.
  Let $v=(v'{\setminus}\{a\})\cup\{b \mid a{\triggers}b\}$, then
  $(v',a,v) \in E_k$ and hence $v\in V_{k+1}$.
  We show that $B\subseteq v$.
  Let $b\in B$.  If $b=a$, then $a\in\uen(s)\cap\outactions$ and hence
  $a$ triggers $a$ in $\II$.  By Prop.~\ref{prop:triggApprox},
  $a\triggers a$ which implies $a\in v$.
  Suppose, instead, that $b\neq a$. If $b\in\uen(s')$, then
  $b\in B'{\setminus}\{a\} \subseteq v'{\setminus}\{a\} \subseteq v$.
  If $b\notin\uen(s')$, then $a$ triggers $b$ in $\II$, and by
  Prop.~\ref{prop:triggApprox}, $a\triggers b$ which implies $b\in v$.
  This proves $B\subseteq v\in\overline{\ES}_\II$ and hence
  $B\in\overline{\ES}_\II$.
  \qed

\end{proof}

The next lemma states that if there is a path from a vertex of
$\EG_\II$ to another vertex, every action in the second vertex is
approximately indirectly triggered by some action in the first vertex.

\begin{lemma}\label{lemm:IndTrig}
  Let $\II$ be a closed \iosa, let $v,v'\in V$ be vertices of
  $\EG_\II$ and let $\rho$ be a path following $E$ from $v$ to
  $v'$. Then for every $b\in v'$ there is an action $a\in v$ such that
  $a\triggers^*b$.
          \end{lemma}
\begin{proof}
  We proceed by induction in the length of $\rho$.  If $|\rho|=0$ then
  $v=v'$ and the lemma holds since $\triggers^*$ is reflexive.
  If $|\rho|>0$, there is a path $\rho'$, $v''\in V$, and
  $c\in\coactions\cap\outactions$ such that
  $\rho=\rho'\cdot(v'',c,v')$.
  By induction, for every action $d\in v''$ there is some $a\in v$
  such that $a\triggers^*d$.
  Because of the definition of $E$ in Def.~\ref{def:egraph}, either $b\in
  v''$ or $c\triggers b$ and $c\in v''$.  The first case follows by
  induction.  In the second case, also by induction, $a\triggers^*c$
  for some $a\in v$ and hence $a\triggers^*b$.
  \qed
                                \end{proof}

The next lemma states that every enabled set $B$ in a composed
\iosa\ is either approximately triggered by a set of initial actions
of the components of the \iosa\ or by a subset of the union of
spontaneously enabled sets in each component where such sets are
spontaneously enabled by the same event.

\begin{lemma}\label{lm:ApproxCommonCause}
  Let $\II=(\II_1\pll\ldots\pll\II_n)$ be a closed \iosa\ and let
  $\{b_1,\ldots,b_m\}\subseteq\coactions\cap\outactions$ be enabled in
  $\II$.
  Then, there are (not necessarily different) $a_1,\ldots,a_m$ such
  that $a_j\triggers^*b_j$, for all $1\leq j\leq m$, and either
  \begin{inparaenum}[(i)]
  \item    $\{a_1,\ldots,a_m\}\subseteq\bigcup_{i=1}^{n}\uen(s^0_i)\cap\outactions_i$,
    or
  \item    there exists $e\in\actions$ and (possibly empty) sets $B_1,\ldots,B_n$
    spontaneously enabled by $e$ in $\II_1,\ldots,\II_n$ respectively,
    such that $\{a_1,\ldots,a_m\}\subseteq\bigcup_{i=1}^n B_i$.
  \end{inparaenum}
                          \end{lemma}
\begin{proof}
  Because of Lemma~\ref{lm:enabledGraphSuper} there is a vertex $v$ of
  $\EG_\II$ such that $\{b_1,\ldots,b_n\}\subseteq v$.
  Because of the inductive construction of $E$ and $V$, there is a
  path from some $v'\in V_0$ to $v$ in $\EG_\II$.
  From Lemma~\ref{lemm:IndTrig}, for each $1\leq j\leq m$, there is an
  $a_j\in v'$ such that $a_j\triggers^*b_j$.
  Because $v'\in V_0$, then either
  $v'=\bigcup_{i=1}^{n}\uen(s^0_i)\cap\outactions_i$ or there is some
  $e\in\actions$ such that $v'=\bigcup_{i=1}^{n} B_i$ with $B_i$
  spontaneously enabled by $e$ in $\II_i$
  \qed
\end{proof}

The following theorem is the main result of this section and provides
sufficient conditions to guarantee that a closed composed \iosa\ is
confluent or, as stated in the theorem, necessary conditions for the \iosa\
to be non-confluent.

\begin{theorem}\label{theo:sufCondForDet}
  Let $\II=(\II_1\pll\cdots\pll\II_n)$ be a closed \iosa.
  If $\II$ potentially reaches a non-confluent state then there are
  actions $a,b\in\coactions\cap\outactions$ such that
        some $\II_i$ is not confluent w.r.t.\ $a$ and $b$, and
      there are $c$ and $d$ such that $c\triggers^*a$, $d\triggers^*b$,
    and, either
    \begin{inparaenum}[(i)]
    \item      $c$ and $d$ are initial actions in any component, or
    \item      there is some $e\in\actions$ and (possibly empty) sets
      $B_1,\ldots,B_n$ spontaneously enabled by $e$ in
      $\II_1,\ldots,\II_n$ respectively, such that
      $c,d\in\bigcup_{i=1}^n B_i$.
    \end{inparaenum}
                                                                                      \end{theorem}
\begin{proof}
  Suppose $\II$ potentially reaches a non confluent state $s$.  Then
  there are necessarily $a,b\in\uen(s)$ that show it and hence $\II$
  is not confluent w.r.t.\ $a$ and $b$.
  By Prop.~\ref{prop:compPreservesConfl}, there is necessarily a
  component $\II_i$ that is not confluent w.r.t.\ $a$ and $b$.
  Since $\{a,b\}$ is an enabled set in $\II$, the rest of the theorem
  follows by Lemma~\ref{lm:ApproxCommonCause}.
  \qed
\end{proof}

Because of Prop.~\ref{prop:only:pot:reach} and
Theorem~\ref{th:iosacdeterministic}, if all potentially reachable
states in a closed \iosa\ $\II$ are confluent, then $\II$ is weakly
deterministic.  Thus, if no pair of actions satisfying conditions in
Theorem~\ref{theo:sufCondForDet} are found in $\II$, then $\II$ is
weakly deterministic.

Notice that the \iosa\ $\II=\II_1\pll\II_2\pll\II_3$ of
Example~\ref{ex:iosas:pll} (see also Figs.~\ref{fig:components}
and~\ref{fig:composed}) is an example that does not meet the
conditions of Theorem~\ref{theo:sufCondForDet}, and hence detected as
confluent.
$c$ and $d$ are the only potential non-confluent actions, which is
noticed in state $s_6$ of $\II_3$.
The approximate indirect
triggering relation can be calculated to
$\triggers^*=\{(c,c),(d,d)\}$.  Also, $\{c\}$ is spontaneously enabled
by $a$ in $\I_1$ and $\{d\}$ is spontaneously enabled by $b$ in
$\I_2$. Since both sets are spontaneously enabled by \emph{different}
actions and $c$ and $d$ are not initial, the set $\{c,d\}$ does not
appear in $V_0$ of $\EG_\II$ which would be required to meet the
conditions of the theorem.

\begin{wrapfigure}[10]{r}{4.2cm}
  \centering\vspace{-2em}
  \begin{tikzpicture}[shorten >=1pt,node distance=1.4cm,on grid,auto, scale=0.8, transform shape, initial text={}, initial where=left]
    \tikzstyle{every state}=[text=black,draw=black,fill=white,minimum size=.2cm]
    \node[state,initial] (s11) {};
    \node[state]         (s12) [right=of s11] {};
    \node[state]         (s13) [right=of s12] {};
    \node[state,draw=white] (I1) [left=1cm of s11] {$\II_1$};

    \node[state,initial] (s21) [below=1cm of s11] {};
    \node[state]         (s22) [right=of s21] {};
    \node[state]         (s23) [right=of s22] {};
    \node[state,draw=white] (I2) [left=1cm of s21] {$\II_2$};

    \node[state,initial] (s31) [below left=1.6cm and .5cm of s21] {};
    \node[state]         (s32) [above right=.6cm and 1cm of s31] {};
    \node[state]         (s33) [below right=.6cm and 1cm of s31] {};
    \node[state]         (s34) [right=of s32] {};
    \node[state]         (s35) [right=of s33] {};
    \node[state]         (s36) [right=of s34] {};
    \node[state,draw=white] (I3) [above left=.6cm and .5cm of s31] {$\II_3$};

    \path[->]
    (s11) edge [] node {$a?$} (s12)
    (s12) edge [] node {$b!!$} (s13)

    (s21) edge [] node {$a?$} (s22)
    (s22) edge [] node {$c!!$} (s23)

    (s31) edge [] node[xshift=.1cm] {$b??$} (s32)
    (s31) edge [] node[yshift=-.5cm,xshift=-.6cm] {$c??$} (s33)
    (s32) edge [] node {$c??$} (s34)
    (s33) edge [] node {$b??$} (s35)
    (s34) edge [] node {$a!$} (s36);
  \end{tikzpicture}
  \caption{$\I_1\pll\I_2\pll\I_3$ meets conditions in
    Theorem~\ref{theo:sufCondForDet}}\label{fig:false:negative}
\end{wrapfigure}

Conditions in Theorem~\ref{theo:sufCondForDet} are not sufficient and
confluent \iosa{s} may satisfy them.  Consider the \iosa{s} in
Fig.~\ref{fig:false:negative}. $\I_1\pll\I_2\pll\I_3$ is a closed
\iosa\ with a single state and no outgoing transition. Hence, it is
confluent.
However, $\I_3$ is not confluent w.r.t.\ $b$ and $c$,
$\triggers^*=\{(b,b),(c,c)\}$, $B_1=\{b\}$ is spontaneously enabled by
$a$ in $\I_1$, and $B_2=\{c\}$ is spontaneously enabled by $a$ in
$\I_2$. Hence $b,c\in\bigcup_{i=1}^n B_i$, thus meeting the conditions
of Theorem~\ref{theo:sufCondForDet}.

\section{Concluding remarks}

In this article, we have extended \iosa\ as introduced
in~\cite{DArgenio2016} with urgent actions.  Though such extension
introduces non-determinism even if the \iosa\ is closed, it does so in a
limited manner.  We were able to characterize when a \iosa\ is weakly
deterministic, which is an important concept since weakly deterministic
\iosa{s} are amenable to discrete event simulation.  In particular,
we showed that closed and confluent \iosa{s} are weakly deterministic
and provided conditions to check compositionally if a closed \iosa\ is
confluent.
Open \iosa{s} are naturally non-deterministic due to input
enabledness:
at any moment of time either two different inputs may be enabled or an
input is enabled jointly with a possible passage of time.  Thus, the
property of non-determinism can only be possible in closed \iosa{s}.
However, Theorem~\ref{theo:sufCondForDet} relates open \iosa{s} to the
concept of weak determinism by providing sufficient properties on open
\iosa{s} whose composition leads to a closed weakly deterministic \iosa.
In addition, we notice that languages like
Modest~\cite{tse/BohnenkampDHK06,fmsd/HahnHHK13,Hartmanns15:phd}, that
have been designed for compositional modelling of complex timed and
stochastic systems, embrace the concept of non-determinism as a
fundamental property.  Thus, ensuring weak determinism on Modest
models using compositional tools like Theorem~\ref{theo:sufCondForDet}
will require significant limitations that may easily boil down to
reduce it to \iosa.  Notwithstanding this observation, we remark that
some translation between \iosa\ and Modest is possible through
Jani~\cite{DBLP:conf/tacas/BuddeDHHJT17}.

Finally, we remark that, though not discussed in this paper, the conditions provided by
Theorem~\ref{theo:sufCondForDet}, can be verified in polynomial time
respect to the size of the components and the number of actions.

 {}
\bibliographystyle{splncs03}

\appendix

\section{Proofs}\label{appendix}

\begin{proof}[of Prop.~\ref{prop:comp:is:closed}]
  The proof of restrictions~\ref{def:iosa:input-and-commit-are-reactive},
  \ref{def:iosa:output-is-generative}, \ref{def:iosa:input-enabled},
  and \ref{def:iosa:input-deterministic} follow by straightforward
  inspection of the rules, considering that $\II_1$ and $\II_2$ also
  satisfy the respective restriction, and doing some case analysis.
  Since $\II_1$ and $\II_2$ are compatible,
  restriction~\ref{def:iosa:clock-control-one-output} also follows by
  inspecting the rules taking into account, in addition, that
  $\II_1$ and $\II_2$ satisfy
  restriction~\ref{def:iosa:input-deterministic}.

  To prove \ref{def:iosa:clock-never-go-off-early} we need to take
  into account that $\enablingck(s_1\pll s_2) =
  \enablingck(s_1)\cup\enablingck(s_2)$ (guaranteed by input
  enabling), and that $\enablingck(s_1)$ and $\enablingck(s_2)$ are
  disjoint sets (guaranteed by compatibility).

  We take $\activeck(s_1\pll s_2)=\activeck_1(s_1)\cup
  \activeck_2(s_2)$ and prove that it satisfies conditions (i)--(iv)
  in~\ref{def:iosa:clock-never-go-off-early}.

  \begin{enumerate}[\rm(i)]
  \item    $\activeck(s_0^1\pll s_0^2) = \activeck_1(s_0^1)\cup \activeck_2(s_0^2)
    \subseteq C_0^1\cup C_0^2 = C_0$.
      \item    $\enablingck(s_1\pll s_2)= \enablingck(s_1)\cup\enablingck(s_2)
    \subseteq \activeck_1(s_1)\cup \activeck_2(s_2)=\activeck(s_1\pll s_2)$.
      \item    Let $s_1\pll s_2$ be stable, then $s_1$ and $s_2$ are stable as
    well (guaranteed by input enabledness).
    Then $\activeck(s_1\pll s_2) = \activeck_1(s_1)\cup \activeck_2(s_2) =
    \enablingck(s_1)\cup\enablingck(s_2) = \enablingck(s_1\pll s_2)$.
      \item    Let $t_1\pll t_2 \trans[C,a,C'] s_1\pll s_2$. We prove by cases
    according to the rules in Table~\ref{tb:parcomp}
    \begin{enumerate}
    \item[(\ref{eq:par:l})]      Let $a\in\actions_1\setminus\actions_2$. Then
      $t_1\trans[C,a,C']s_1$ and $s_2=t_2$, and we can calculate:
      ${\activeck(s_1\pll s_2)} = {\activeck_1(s_1)\cup \activeck_2(s_2)}
      = {\activeck_1(s_1)\cup \activeck_2(t_2)} \subseteq
      {(\activeck_1(s_1)\setminus C)\cup C' \cup \activeck_2(t_2)} =
      {((\activeck_1(t_1)\cup \activeck_2(t_2))\setminus C)\cup C'} =
      {(\activeck(t_1\pll t_2)\setminus C)\cup C'}$. In particular, the
      last but one equality follow by compatibility.
          \item[(\ref{eq:par:r})]      Similar to the previous case if $a\in\actions_2\setminus\actions_1$.
          \item[(\ref{eq:par:s})]      Let $a\in\actions_1\cup\actions_2$.  Then
      $t_1\trans[C_1,a,C'_1]s_1$ and $t_2\trans[C_2,a,C'_2]s_2$, with
      $C=C_1\cup C_2$ and $C'=C'_1\cup C'_2$, and we can calculate:
      ${\activeck(s_1\pll s_2)} = {\activeck_1(s_1)\cup \activeck_2(s_2)}
      \subseteq {((\activeck_1(t_1)\setminus C_1)\cup C'_1) \cup ((\activeck_2(t_2)\setminus C_2)\cup C'_2)} =
      {((\activeck_1(t_1)\cup \activeck_2(t_2)) \setminus C_1\cup C_2)\cup C'_1\cup C'_2} =
      {(\activeck(t_1\pll t_2)\setminus C)\cup C'}$.  The last but one
      equality follow by compatibility.
      \qed
    \end{enumerate}
  \end{enumerate}
\end{proof}

\begin{proof}[of Prop.~\ref{prop:compPreservesConfl}]
  Let $s_1\pll s_2$ in $\states_{\II_1\pll\II_2}$, such that
  $s_1\pll s_2\trans[\emptyset,a,C']s'_1\pll s'_2$ and
  $s_1\pll s_2\trans[\emptyset,b,C'']s''_1\pll s''_2$ with
  $a,b\in\actions_{\II_1\pll\II_2}$.
  We proceed by case analysis on each possible combinations of the
  rules in Table~\ref{tb:parcomp} that originates the transitions.
  We prove the case in which $s_1\pll
  s_2\trans[\emptyset,a,C']s'_1\pll s'_2$ is produced by rule $\text{(R1)}$,
  hence $a\in \actions_{\I_1}\setminus\actions_{\I_2}$. The rest
  proceeds in a similar way.
  Then $s'_2=s_2$ and $s_1\trans[\emptyset,a,C']s'_1$.
  We have then three sub-cases given the nature of $b$:
  \begin{itemize}
  \item    If $b\in \actions_{\I_1}{\setminus}\actions_{\I_2}$, rule $\text{(R1)}$
    applies and hence $s''_2=s_2$ and
    $s_1\trans[\emptyset,b,C']s''_1$.  Since $\I_1$ is confluent,
    there exists $s'''_1$ such that $s'_1\trans[\emptyset,a,C']s'''_1$
    and $s''_1\trans[\emptyset,b,C'']s'''_1$. Using $\text{(R1)}$ in both
    cases, $s'_1\pll s_2\trans[\emptyset,a,C']s'''_1\pll s_2$ and
    $s''_1\pll s_2\trans[\emptyset,b,C'']s'''_1\pll s_2$, which proves
    this case.
  \item    If $b\in\actions_{\I_2}{\setminus}\actions_{\I_1}$, $\text{(R2)}$ applies
    and hence $s_1=s''_1$ and $s_2\trans[\emptyset,b,C'']s''_2$.  By
    $\text{(R1)}$, $s_1\pll s''_2\trans[\emptyset,a,C'] s'_1\pll s''_2$, and
    by $\text{(R2)}$, $s'_1\pll s_2\trans[\emptyset,b,C'']s'_1\pll s''_2$
    which proves this case.
  \item    If $b\in\actions_{\I_1}{\cap}\actions_{\I_2}$, $\text{(R3)}$ applies.
    Hence there are $C''_1$ and $C''_2$ such that $C''=C''_1\cup
    C''_2$, $s_1\trans[\emptyset,b,C''_1]s''_1$ and
    $s_2\trans[\emptyset,b,C''_2]s''_2$.  Furthermore, since $\I_1$ is
    confluent, there exists $s'''_1$ such that
    $s'_1\trans[\emptyset,b,C''_1]s'''_1$ and
    $s''_1\trans[\emptyset,a,C']s'''_1$.  Then, by $\text{(R3)}$, $s'_1\pll
    s_2\trans[\emptyset,b,C'']s'''_1\pll s''_2$, and by $\text{(R1)}$,
    $s''_1\pll s''_2\trans[\emptyset,a,C']s'''_1\pll s''_2$, which
    concludes the proof.
    \qed
  \end{itemize}
\end{proof}

\begin{proof}[of Corollary~\ref{cor:weak:transition}]
  Suppose $(s,\vec{v})\Trans[]\mu$ and
  $(s,\vec{v})\Trans[]\mu'$. Then, there are $n_1$, $n_2$, $C_1$ and
  $C_2$, such that $(s,\vec{v})\Trans[C_1]_{n_1}\mu$ and
  $(s,\vec{v})\Trans[C_2]_{n_2}\mu'$.
  By 1.(ii) in Lemma~\ref{lm:weak:transition} there are $s_1$ and
  $s_2$ stables such that $(s,\emptyset,0)\reduct[*](s_1,C_1,n_1)$ and
  $(s,\emptyset,0)\reduct[*](s_2,C_2,n_2)$.
  Since both $s_1$ and $s_2$ are stable, by
  Prop.~\ref{prop:uarnormandconfl}, $(s_1,C_1,n_1)$ and
  $(s_2,C_2,n_2)$ are in normal form, and since they must be unique
  $s_1=s_2$, $C_1=C_2$, and $n_1=n_2$.
  Finally, By 1.(iii) in Lemma~\ref{lm:weak:transition}, $\mu_1=\mu_2$
  \qed
\end{proof}

\begin{proof}[of Lemma~\ref{lm:weak:transition}]
  We proceed by induction on $n$ proving first item 1 and using it to
  prove 2.

  So, suppose $n=1$ and $(s,\vec{v})\Trans[C]_1\mu$. By rule
  $\text{(T1)}$ in Def.~\ref{def:weaktransition}, there exists $s'$
  stable such that $s\trans[\emptyset,\act,C]s'$ for some
  $\act\in\coactions$ with $\mu=\mu^{\vec{v}}_{C,s'}$, which proves (i).
  From here and Def.~\ref{def:urgentabstractreduction},
  $(s,C',m)\reduct[*](s',C'{\cup}C,m{+}1)$, proving (ii).
                  To prove (iii), suppose $(s,\vec{v}')\Trans[C']_1\mu'$. By (i) and
  (ii) applied to this other transition, there exists a stable $s''$
  such that $\mu'=\mu^{\vec{v}'}_{C',s''}$ and
  $(s,\emptyset,0)\reduct[*](s'',C',1)$. But also
  $(s,\emptyset,0)\reduct[*](s',C,1)$ as proven before. Since $s'$ and
  $s''$ are stable, then, by Prop.~\ref{prop:uarnormandconfl}, both
  $(s',C,1)$ and $(s'',C',1)$ are in normal form which must also be
  unique. Then $s'=s''$ and $C'=C''$. Moreover, if $\vec{v}'=\vec{v}$
  then $\mu'=\mu^{\vec{v}'}_{C',s''}=\mu^{\vec{v}}_{C,s'}=\mu$.

  To prove item 2 for $n=1$, notice first that, by (iii), $f^C_1$ is
  indeed a function.  By (i), $f^C_1(t,\vec{w})=\mu^{\vec{w}}_{C,t'}$
  whenever $(t,\vec{w})\Trans[C]_1\mu^{\vec{w}}_{C,t'}$ for some $t'$
  stable which is granted to exist, and $f^C_1(t,\vec{w})=\emptyfun$
  otherwise.  To show that $f^C_1$ is measurable,
  by~\cite[Lemma~3.6]{Viglizzo2005:phd}, it suffices to prove that
  $(f^C_1)^{-1}(\Delta^q(A\times\prod_{i=1}^NV_i))$ is measurable for
  all $A\subseteq\states$ and $V_i\in\borel(\R)$.  Notice that
  \begin{align*}
    & (f^C_1)^{=1}(\Delta^q(A\times{\textstyle\prod_{i=1}^NV_i})) = \\
    & = \{(t,\vec{w})\mid \exists t':
    (t,\vec{w})\Trans[C]_1\mu^{\vec{w}}_{C,t'} \land
    \mu^{\vec{w}}_{C,t'}(A\times{\textstyle\prod_{i=1}^NV_i})\geq q\} \\
    & = \{(t,\vec{w})\mid \exists t'{\in}A:
    (t,\vec{w})\Trans[C]_1\mu^{\vec{w}}_{C,t'} \land
    {\textstyle\underset{x_i\in C}{\prod} \mu_{x_i}}({\textstyle\underset{x_i\in C}{\prod} V_i})\geq q \land
    \forall x_i\notin C : \vec{w}(i)\in V_i \}   \end{align*}

  \begin{align*}
    & = {\textstyle \bigcup_{\substack{t\in\states\\t'\in A}}}
    \underbrace{
      \{(t,\vec{w})\mid
      (t,\vec{w})\Trans[C]_1\mu^{\vec{w}}_{C,t'} \land
      {\textstyle\underset{x_i\in C}{\prod} \mu_{x_i}}({\textstyle\underset{x_i\in C}{\prod} V_i})\geq q \land
      \forall x_i\notin C : \vec{w}(i)\in V_i \}
    }_{\normalsize = X_t}
  \end{align*}
  Notice that, if ${\textstyle\prod_{x_i\in C}
    \mu_{x_i}}({\textstyle\prod_{x_i\in C} V_i})\geq q$, then
  $X_t=\{t\}\times{\textstyle\prod_{i=1}^N\overline{V}_i}$, with
  $\overline{V}_i=\R$ if $x_i\in C$ and $\overline{V}_i=V_i$ if
  $x_i\notin C$, and $X_t=\emptyset$ otherwise.  In both cases $X_t$
  is measurable.
  Since $\states$ is finite, the union is also finite and hence
  $f^C_1$ es measurable, which proves the base case.

  For the inductive case, let $n\geq 1$ and suppose
  $(s,\vec{v})\Trans[C]_{n+1}\mu$.  By $\text{(T2)}$, there are $C'$
  and $C''$ such that $C=C'\cup C''$, $s\trans[\emptyset,\act,C']s'$,
  $\forall \vec{v}'\in\R^\N: (s',\vec{v}')\Trans[C'']_n\mu'$, and
  $\mu=\int_{\states\times\R^N} f_n^{C''}d\mu^{\vec{v}}_{C',s'}$.
  By induction, $C''$ is unique (by 1.(iii)),
  $(s',\vec{v}')\Trans[C'']_n\mu^{\vec{v}'}_{C'',s''}$ for all
  $\vec{v'}$ and unique stable state $s''$ (by 1.(i) and 1.(ii)), and
  $f_n^{C''}$ is measurable (by 2).  Thus $\int_{\states\times\R^N}
  f_n^{C''}d\mu^{\vec{v}}_{C',s'}$ is well defined.  Moreover, notice
  that $f_n^{C''}(s',\vec{v}')=\mu^{\vec{v}'}_{C'',s''}$ for all
  $\vec{v}'$.

  We focus on 1.(i) and show that $\mu=\mu^{\vec{v}}_{C'\cup C'',s''}$.
  First, notice that
  $\mu=\int_{\{s'\}\times\R^N}f_n^{C''}d\mu^{\vec{v}}_{C',s'}+\int_{(\states\setminus\{s'\})\times\R^N}f_n^{C''}d\mu^{\vec{v}}_{C',s'}$
  and since
  $\mu^{\vec{v}}_{C',s'}=\delta_{s'}\times\prod_{i=1}^N\overline{\mu}^{v_i}_{x_i}$
  with $\overline{\mu}^{v_i}_{x_i} = {\mu}_{x_i}$ if $x_i \in C'$ and
  $\overline{\mu}^{v_i}_{x_i} = \delta_{v_i}$ otherwise (we write
  $v_i$ for $\vec{v}(i)$), then the second summand is the null
  function $\emptyfun$.
  Now, for $A\subseteq\states$ and $Q_i\in\R$, $1\leq i\leq N$, we
  calculate
  \begin{align*}
    \mu(A\times Q_1\times\cdots\times Q_N)
    & = \int_{\{s'\}\times\R^N}f_n^{C''}(t,\vec{w})(A\times Q_1\times\cdots\times Q_N)\ d\mu^{\vec{v}}_{C',s'}(t,\vec{w}) \\
    & = \int_{\R^N}f_n^{C''}(s',\vec{w})(A\times Q_1\times\cdots\times Q_N)\ d({\textstyle\prod_{i=1}^N}\overline{\mu}^{v_i}_{x_i})(\vec{w}) \\
    & = \int_{\R^N}\mu^{\vec{w}}_{C'',s''}(A\times Q_1\times\cdots\times Q_N)\ d({\textstyle\prod_{i=1}^N}\overline{\mu}^{v_i}_{x_i})(\vec{w}) = (\dag)
  \end{align*}
  By definition,
  $\mu^{\vec{w}}_{C'',s''}=\delta_{s''}\times\prod_{i=1}^N\overline{\mu}^{w_i}_{x_i}$
  with $\overline{\mu}^{w_i}_{x_i} = {\mu}_{x_i}$ if $x_i \in C''$ and
  $\overline{\mu}^{w_i}_{x_i} = \delta_{v_i}$ otherwise. Then (in the
  following we omit the domain of each integral is $\R$), using
  Fubini's theorem, we have:
  \begin{align*}
    &    (\dag)
    = \compactints\delta_{s''}(A)\cdot\overline{\mu}^{w_1}_{x_1}(Q_1)\compactcdots\overline{\mu}^{w_N}_{x_N}(Q_N)\ d\overline{\mu}^{v_1}_{x_1}(w_1)...d\overline{\mu}^{v_N}_{x_N}(w_N) \\
    &    = \delta_{s''}(A)\!\compactints\!\overline{\mu}^{w_2}_{x_2}(Q_2)\compactcdots\overline{\mu}^{w_N}_{x_N}(Q_N)\Big(\!\underbrace{\int\!\overline{\mu}^{w_1}_{x_1}(Q_1)\,d\overline{\mu}^{v_1}_{x_1}(w_1)}_{(*)}\!\Big) d\overline{\mu}^{v_2}_{x_2}(w_2)...d\overline{\mu}^{v_N}_{x_N}(w_N)
  \end{align*}
  We focus on $(*)$. Three cases may arise.
  If $x_1\in C''$, then
  $(*) = \int\mu_{x_1}(Q_1)\,d\overline{\mu}^{v_1}_{x_1}(w_1) =
  \mu_{x_1}(Q_1)\int d\overline{\mu}^{v_1}_{x_1}(w_1)=\mu_{x_1}(Q_1)$
  since $\int d\overline{\mu}^{v_1}_{x_1}(w_1)=1$.
  If $x_1\in C'\setminus C''$,
  $(*) = \int\delta_{w_1}(Q_1)\,d\mu_{x_1}(w_1) =
  \int\chi_{Q_1}(w_1)\,d\mu_{x_1}(w_1) = \mu_{x_1}(Q_1)$
  where $\chi_{Q_1}$ is the usual characteristic function.
  Finally, if $x_1\notin C\cup C''$,
  $(*) = \int\delta_{w_1}(Q_1)\,d\delta_{v_1}(w_1) =
  \int\chi_{Q_1}(w_1)\,d\delta_{v_1}(w_1) = \delta_{v_1}(Q_1)$.
  Therefore $(*)=\overline{\mu}_{x_1}(Q_1)$ with
  $\overline{\mu}_{x_1}=\mu_{x_1}$ if $x_1\in C'\cup C''$ and
  $\overline{\mu}_{x_1}=\delta_{v_1}$ otherwise.
  Then, proceeding in the same manner for all the indices, we continue,
  \begin{align*}
    &= \delta_{s''}(A)\overline{\mu}_{x_1}(Q_1)\!\compactints\!\overline{\mu}^{w_2}_{x_2}(Q_2)\compactcdots\overline{\mu}^{w_N}_{x_N}(Q_N)\, d\overline{\mu}^{v_2}_{x_2}(w_2)...d\overline{\mu}^{v_N}_{x_N}(w_N)\\
    &= \delta_{s''}(A)\cdot\overline{\mu}_{x_1}(Q_1)\cdots\overline{\mu}_{x_N}(Q_N)
    = (\delta_{s''}\times{\textstyle\prod_{i=1}^{N}}\overline{\mu}_{x_i})(A\times Q_1\times\cdots\times Q_N) \\
    &= \mu^{\vec{v}}_{C\cup C'',s''}(A\times Q_1\times\cdots\times Q_N)
  \end{align*}
  which proves 1.(i).

  To prove 1.(ii), by Def.~\ref{def:urgentabstractreduction},
  $(s,C^*,m)\reduct[](s',C^*{\cup}C',m{+}1)$ since
  $s\trans[\emptyset,\act,C']s'$.  By induction,
  $(s',\vec{v}')\Trans[C'']_n\mu'$ implies
  $(s',C^*{\cup}C',m{+}1)\reduct[*](s'',C^*{\cup}C'{\cup}C'',m{+}1{+}n)$.
  Thus $(s,C^*,m)\reduct[*](s'',C^*{\cup}C'{\cup}C'',m{+}1{+}n)$,
  which proves 1.(ii).

  The proofs of 1.(iii) and 2 follows like for the base case.
  \qed
\end{proof}

\begin{proof}[of Lemma~\ref{lemma:one-step-invariant}]
  We proceed analyzing by cases according $a$ is $\init$, in
  $\actions$, or in $\R_{>0}$.

  If $a$ is $\init$, we only consider cases where $s=\sinit$, since
  $\transitions_{\init}(s,v)=\emptyset$ otherwise.
  If $\mu\in\transitions_{\init}(\init,v)$, then
  $\mu=\delta_{s_0}\times\prod_{i=1}^N\mu_{x_i}$.  Since each
  $\mu_{x_i}$ is a continuous probability measure, the likelihood of
  two clocks being set to the same value is $0$ and
  $\mu_{x_i}(\R_{>0})=1$. Then $\mu(\notinv)=0$.
  This proves the first case.

  For the other cases we introduce the following notation.
  For each $x_i,x_j\in\activeck(s')$, define
  $\notinv_{ij}=\{(s'',\vec{w})\mid \vec{w}(i)=\vec{w}(j)\}$
  whenever $i\neq j$,
  $\notinv_{i,\text{st}}=\{(s'',\vec{w})\mid\stable{s''}, \vec{w}(i)<0\}$, and
  $\notinv_{i,\text{nst}}=\{(s'',\vec{w})\mid\neg\stable{s''},\vec{w}(i)\leq0\}$.
  It is not difficult to prove that each of this type of sets is
  measurable.
  Notice that
  $\notinv=\bigcup\notinv_{ij}\cup\bigcup\notinv_{i,\text{st}}\cup\bigcup\notinv_{i,\text{nst}}$
  and, since the unions are finite, $\mu(\notinv)=0$ if and only if
  $\mu(\notinv_{ij})=0$, $\mu(\notinv_{i,\text{st}})=0$, and
  $\mu(\notinv_{i,\text{nst}})=0$, for every $i,j$.  Thus, for the
  remaining two cases we focus on proving these last three equalities.

  Let $a\in\actions$, $\mu\in\transitions_a(s,\vec{v})$ and
  $(s,\vec{v})\in\inv$.  Then $s\neq\init$ and hence, by
  Def.~\ref{def:iosasemantic}, there exists $s\trans[C, \act, C']s'$
  such that $\bigwedge_{x_i \in C} \vec{v}(i)\leq 0$, and
  $\mu=\delta_{s'}\times\prod_{i=1}^N\overline{\mu}_{x_i}$ with
  $\overline{\mu}_{x_i}=\mu_{x_i}$ if $x_i\in C$,
  $\overline{\mu}_{x_i}=\delta_{\vec{v}(i)}$ otherwise.

  Let $x_i\in\activeck(s')$, then
  $x_i\in(\activeck(s)\setminus C)\cup C'$.
  If $x_i\in C'$, then $\mu_{x_i}(\R_{>0})=1$ and hence
  $\mu(\notinv_{i,\text{st}})=\mu(\notinv_{i,\text{nst}})=0$.
  If $x_i\in (\activeck(s)\setminus C)\setminus C'$ we consider two
  subcases: either $C=\emptyset$ or $C=\{x_j\}$.
  In the first case, $a\in\coactions$ and therefore $s$ is not
  stable. Then $\vec{v}(i)>0$ (since $(s,\vec{v})\in\inv$) and hence
  $\delta_{\vec{v}(i)}(\R_{>0})=1$, which implies
  $\mu(\notinv_{i,\text{st}})=\mu(\notinv_{i,\text{nst}})=0$.
  If instead $C=\{x_j\}$, $i\neq j$ and, by
  Def.~\ref{def:iosasemantic}, $\vec{v}(j)=0$. Since $s$ is stable and
  $(s,\vec{v})\in\inv$, then $\vec{v}(i)\geq0$ and
  $\vec{v}(i)\neq\vec{v}(j)$, hence $\vec{v}(i)>0$ and, as before,
  $\mu(\notinv_{i,\text{st}})=\mu(\notinv_{i,\text{nst}})=0$.

  Suppose now $x_i,x_j\in\activeck(s')$ with $i\neq j$, then
  $x_i,x_j\in(\activeck(s)\setminus C)\cup C'$. If $x_i\in C$ then
  $\mu_{x_i}$ is a continuous probability measure and hence
  $\mu(\notinv_{ij})=0$.  Similarly if $x_j\in C$.  If instead
  $x_i,x_j\in\activeck(s){\setminus}C$, then
  $\vec{v}(i)\neq\vec{v}(j)$ because $(s,\vec{v})\in\inv$ and hence
  $\delta_{\vec{v}(i)}\neq\delta_{\vec{v}(j)}$. Therefore
  $\mu(\notinv_{ij})=0$.
  This proves that $\mu(\notinv)=0$ for this case.

  Finally, take $d\in\R_{>0}$ and suppose that
  $\transitions_d(s,\vec{v})=\{\mu\}$ with
  $(s,\vec{v})\in\inv$.
  By Def.~\ref{def:iosasemantic}, $s$ needs to be stable,
  $0<d\leq\min\{\vec{v}(k) \mid s\trans[\{x_k\},a,C']s', a{\in}\outactions\}$,
  and
  $\mu=\delta_{s}\times\prod_{i=1}^N\delta_{\vec{v}(i)-d}$.
  Since $s$ is stable, $\mu(\notinv_{i,\text{nst}})=0$.
  For $x_i\in\activeck(s)$,
  ${\vec{v}(i){-}d}
  \geq {\min\{\vec{v}(k) \mid {s\trans[\{x_k\},a,C']s'}, {a{\in}\actions^O}\}-d}
  \geq 0$,
  since $\activeck(s)=\enablingck(s)$ ($s$ is stable).  Hence
  $\delta_{\vec{v}(i)-d}(\R_{\geq 0})=1$. Therefore
  $\mu(\notinv_{i,\text{st}})=0$.
  For $x_i,x_j\in\activeck(s)$ with $i{\neq}j$,
  $\vec{v}(i)\neq\vec{v}(j)$ because $(s,\vec{v})\in\inv$.  Hence
  $\delta_{\vec{v}(i)-d}\neq\delta_{\vec{v}(j)-d}$.  So
  $\mu(\notinv_{ij})=0$.  This proves that $\mu(\notinv)=0$ for this
  case, and therefore the lemma.
  \qed
\end{proof}

\begin{proof}[of Lemma~\ref{lemma:streach}]
  We only prove it for $\II_1\pll\II_2$. The generalization to any $n$
  follows easily.
  We prove it by induction on the length of the plausible path $\sigma$ that leads
  to $s_1\pll s_2$.  If $|\sigma|=0$ the $\sigma=s^0_1\pll s^0_2$,
  where each $s^0_i$ is initial in each $\II_i$ and hence potentially
  reachable.
  For the inductive case let
  $\sigma=\sigma'\cdot(s'_1\pll s'_2)\trans[C,a,C'](s_1\pll s_2)$.
  W.l.o.g.\ and by contradiction, suppose $s_1$ is not potentially
  reachable in $\II_1$.  Necessarily, $s_1\neq s'_1$ since $s'_1$ is
  potentially reachable by induction ($|\sigma|=|\sigma'|+1$).  Thus
  $s'_1\pll s'_2\trans[C,a,C']s_1\pll s_2$ is the result of applying
  $\textrm{(R1)}$ or $\textrm{(R3)}$.  The rest of the proof follows
  similarly for both cases.  So suppose $\textrm{(R3)}$ was applied.
  Then $s'_1\trans[C_1,a,C'_1]s_1$ for some $C_1\subseteq C$ and
  $C'_1\subseteq C'$.  Since $s_1$ is not potentially reachable but
  $s'_1$ is, then $a\in\actions\setminus\coactions$ and there is a
  $b\in\coactions\cap\outactions$ such that
  $s'_1\trans[\_,b,\_]\_$. Then $s'_1\pll s'_2\trans[\_,b,\_]\_$,
  either by $\textrm{(R1)}$ or by $\textrm{(R3)}$ (being $\II_2$ input
  enabled) yielding $\sigma$ not plausible and hence a contradiction.
                                                      \qed
\end{proof}

\begin{proof}[of Theorem~\ref{th:iosacdeterministic}]
  We have to show that every measurable set $B\in\borel(\pstates)$ of
  states satisfying conditions (a), (b), or (c) in
  Def.~\ref{def:deterministiciosa} is almost never reached in
  $\P(\I)$.
  Let
  $B_{\text{st}}=B\cap((\{s\mid\stable{s}\}\cup\{\init\})\times\R^N)$
  and $B_{\neg\text{st}}=B\cap (\{s\mid\neg\stable{s}\}\times\R^N)$.
  Then $B=B_{\text{st}}\cup B_{\neg\text{st}}$, and $B_{\text{st}}$ and
  $B_{\neg\text{st}}$ are measurable.  Hence $B$ is almost never reached if
  and only if $B_{\text{st}}$ and $B_{\neg\text{st}}$ are almost never
  reached.

  Let $\mathsf{En}_{\geq 2} = \{(s,\vec{v})\in\pstates \mid
  (\stable{s} \lor s=\init) \land
  {|\bigcup_{\act\in\actions\cup\{\init\}}\transitions_\act(s,\vec{v})|}\geq2\}$.
  By Lemma~\ref{lemma:inv-enables-at-most-one},
  $\mathsf{En}_{\geq2}\subseteq\notinv$, and by (a) in
  Def.~\ref{def:deterministiciosa},
  $B_{\text{st}}\subseteq\mathsf{En}_{\geq 2}$.  Then, by
  Corollary~\ref{cor:invariant}, $B_{\text{st}}$ is almost never
  reached.
  In addition, Corollary~\ref{cor:weak:transition}, ensures that no
  $(s,\vec{v})\in B_{\neg\text{st}}$ satisfies (b).  Therefore every
  $(s,\vec{v})\in B_{\neg\text{st}}$ satisfies (c).  Hence, by
  Lemma~\ref{lem:nonurgentunreachable}
  $B_{\neg\text{st}}\subseteq\notinv$.  Then, by
  Corollary~\ref{cor:invariant}, $B_{\neg\text{st}}$ is almost never
  reached, which proves the theorem.
  \qed
\end{proof}

\end{document}